\documentclass[11pt]{article}
\usepackage{amssymb}
\usepackage{amsmath}
\usepackage{amsthm}
\usepackage{latexsym}
\usepackage{hyperref}
\usepackage{mathdots}
\usepackage[all]{xy}
\usepackage{stmaryrd}
\usepackage[normalem]{ulem}
\usepackage[capitalise]{cleveref}
\usepackage{tikz}
\usetikzlibrary{patterns}
\usepackage{authblk}
\theoremstyle{definition}
\newtheorem{theo}{Theorem}[section]
\newtheorem{defi}[theo]{Definition}
\newtheorem{prop}[theo]{Proposition}
\newtheorem{defiprop}[theo]{Definition/Proposition}
\newtheorem{cor}[theo]{Corollary}
\newtheorem{lemma}[theo]{Lemma}
\newtheorem{exa}[theo]{Example}
\newtheorem{rem}[theo]{Remark}

\numberwithin{equation}{section}

\newcommand{\DS}{\displaystyle}

\newcommand{\N}{{\mathbb N}}
\newcommand{\F}{{\mathbb F}}

\newcommand{\Q}{{\mathbb Q}}

\newcommand{\cA}{{\mathcal A}}
\newcommand{\cC}{{\mathcal C}}

\newcommand{\cF}{{\mathcal F}}

\newcommand{\cM}{{\mathcal M}}

\newcommand{\cB}{{\mathcal B}}

\newcommand{\cS}{{\mathcal S}}

\newcommand{\cU}{{\mathcal U}}
\newcommand{\cV}{{\mathcal V}}

\newcommand{\T}{\mbox{$^{\sf T}$}}

\newcommand{\subspace}[1]{\mbox{$\langle{#1}\rangle$}}
\newcommand{\inner}[2]{\mbox{$\langle{#1}\mid{#2}\rangle$}}

\newcommand{\Hom}{\mbox{\rm Hom}}
\newcommand{\Fnm}{\mbox{$\F^{n\times m}$}}

\newcommand{\rk}{{\rm rk}}
\newcommand{\rowrk}{\mbox{$\text{\rm rowrk}_{\F}$}}
\newcommand{\drk}{\mbox{${\rm d}_{\rm rk}$}}
\newcommand{\maxrk}{{\rm maxrk}}
\newcommand{\dd}{\mbox{${\rm d}$}}
\newcommand{\tr}{{\rm tr}}
\newcommand{\rr}{{\rm r}}
\newcommand{\cc}{{\rm c}}
\newcommand{\cl}{{\rm cl}}
\newcommand{\GL}{{\rm GL}}
\newcommand{\rowsp}{\mbox{\rm rowsp}}
\newcommand{\colsp}{\mbox{\rm colsp}}

\newcommand{\mmid}{\!\mid\!}

\newcommand{\qPM}{$q$-PM}
\newcommand\pperp{\protect\mathpalette{\protect\independenT}{\perp}}
\def\independenT#1#2{\mathrel{\rlap{$#1#2$}\mkern2mu{#1#2}}}
\newcommand{\Smalltwomat}[2]{\mbox{$\left(\begin{smallmatrix}{#1}\\[.4ex]{#2}\end{smallmatrix}\right)$}}

\newcounter{alp}
\newcounter{ara}
\newcounter{rom}

\newenvironment{alphalist}{\begin{list}{(\alph{alp})\hfill}{\usecounter{alp}
     \topsep-1.4ex \labelwidth.7cm \leftmargin.7cm \labelsep0cm
     \rightmargin0cm \parsep0ex \itemsep.5ex
     \partopsep-1.4ex}}{\end{list}}
\newenvironment{arabiclist}{\begin{list}{(\arabic{ara})\hfill}{\usecounter{ara}
     \topsep-1ex \labelwidth.7cm \leftmargin.7cm \labelsep0cm
     \rightmargin0cm \parsep0ex \itemsep.5ex
     \partopsep-1ex}}{\end{list}}

\newenvironment{mylist2}{\begin{list}{(\arabic{ara})\hfill}{\usecounter{ara}
     \topsep-1ex \labelwidth1cm \leftmargin1cm \labelsep0cm
     \rightmargin0cm \parsep0ex \itemsep.5ex
     \partopsep-1ex}}{\end{list}}
\newenvironment{mylist3}{\begin{list}{(\arabic{ara})\hfill}{\usecounter{ara}
     \topsep-1ex \labelwidth1.2cm \leftmargin1.2cm \labelsep0cm
     \rightmargin0cm \parsep0ex \itemsep.5ex
     \partopsep-1ex}}{\end{list}}

\makeatletter
\let\@fnsymbol\@arabic
\makeatother
\begin{document}

\title{$q$-Polymatroids and Their Relation to Rank-Metric Codes}
\author{Heide Gluesing-Luerssen\thanks{Corresponding author. Department of Mathematics, University of Kentucky, Lexington KY 40506-0027, USA; heide.gl@uky.edu. HGL was partially supported by the grant \#422479 from the Simons Foundation.}\quad and Benjamin Jany\thanks{Department of Mathematics, University of Kentucky, Lexington KY 40506-0027, USA; benjamin.jany@uky.edu.}}

\date{March 10, 2022}
\maketitle
	
\begin{abstract}
\noindent It is well known that linear rank-metric codes give rise to $q$-polymatroids.
Analogously to matroid theory one may ask whether a given $q$-polymatroid is representable by a rank-metric code.
We provide an answer by presenting an example of a $q$-matroid that is not representable by any linear rank-metric code and, via a relation to paving matroids, provide examples of various $q$-matroids that are not representable by $\F_{q^m}$-linear rank-metric codes.
We then go on and introduce deletion and contraction for $q$-polymatroids and show that they are mutually dual and
correspond to puncturing and shortening of rank-metric codes.
Finally, we introduce a closure operator along with the notion of flats and show that the generalized rank weights of a
rank-metric code are fully determined by the flats of the associated $q$-polymatroid.
\end{abstract}

\textbf{Keywords:} Rank-metric codes, $q$-matroids, $q$-polymatroids, representability.

\section{Introduction}\label{S-Intro}
Rank-metric codes -- originally introduced by Delsarte~\cite{Del78} and later independently re-discovered by
Gabidulin~\cite{Gab85} as well as Roth~\cite{Ro91} -- have been in the focus of algebraic coding theory
throughout the last 15 years thanks to their suitability for communication networks.
Their coding-theoretic properties have been studied in detail, and various constructions of optimal codes, such as MRD codes,
have been found. For details we refer to the vast literature.

In this paper we focus on the algebraic and combinatorial aspects of rank-metric codes and study them with the aid of associated $q$-polymatroids.
We will focus on linear rank-metric codes, that is, subspaces of some matrix space $\F_q^{n\times m}$, endowed with
the rank metric.
On various occasions we will consider $\F_{q^m}$-linear rank-metric codes, that is,
codes that turn into $\F_{q^m}$-subspaces of $\F_{q^m}^n$ under a suitable identification of
$\F_q^{n\times m}$ with $\F_{q^m}^n$.
Not surprisingly, the algebraic and combinatorial properties of a rank-metric code depend on the `degree of linearity'.

In~\cite{JuPe18} Jurrius/Pellikaan introduce $q$-matroids and show that $\F_{q^m}$-linear rank-metric codes give rise to $q$-matroids, thus providing a vast variety of examples of $q$-matroids.
As the terminology indicates, $q$-matroids  form the $q$-analogue of matroids: instead of subsets of a finite set one considers subspaces of a finite-dimensional vector space over a finite field.
Furthermore, Jurrius/Pellikaan present various cryptomorphic definitions of $q$-matroids.
In~\cite{BCJ22} (and the precursor~\cite{BCIJS20}), Byrne and co-authors considerably extend the list of
cryptomorphic definitions.
As has been shown in \cite{JuPe18,BCIJS20,BCJ22}, the theory of $q$-matroids nicely parallels the theory of matroids.
It should be noted that $q$-matroids appeared already much earlier in the Ph.D.~thesis \cite{Cra64} but remained unnoticed in the coding community until~\cite{JuPe18}.

While $\F_{q^m}$-linear rank-metric codes give rise to $q$-matroids, this is not the case for $\F_q$-linear rank-metric codes.
However, as shown by Gorla and co-authors in~\cite{GJLR19}  as well as Shiromoto~\cite{Shi19} and
Ghorpade/Johnson~\cite{GhJo20}, $\F_q$-linear rank-metric codes induce $q$-polymatroids.
This means that the rank function attains rational values.
As for classical polymatroids this seemingly slight generality in the rank function causes $q$-polymatroids to be much less rigid than
$q$-matroids.
An even further generalization appears in~\cite{BMS20}, where Britz and co-authors study $q$-demimatroids associated with rank-metric codes.

In this paper we will make further contributions to the theory of $q$-polymatroids.
Different from \cite{GJLR19,GhJo20,Shi19} we will study $q$-polymatroids over general $n$-dimensional $\F_q$-vector spaces~$E$
rather than~$\F_q^n$.
This forces us to revisit a duality result from~\cite{GJLR19} and show that the equivalence class of the
dual $q$-polymatroid does not depend on the choice of the non-degenerate symmetric bilinear form on~$E$.
The purpose of this slight generalization becomes clear only when we study deletions and contractions as the latter naturally
lead to general ground spaces.
Special care in the choice of the bilinear forms is needed to show that deletion and contraction are
mutually dual (up to equivalence).
With all notions properly in place, we then show that deletion and contraction correspond to puncturing and shortening
of rank-metric codes.

Naturally, one may wonder whether a given $q$-polymatroid is representable in the sense that it arises from a rank-metric code.
Making use of non-representable (classical) matroids, we present some $q$-matroids that are not representable by
$\F_{q^m}$-linear rank-metric codes, thereby answering an open problem in~\cite{JuPe18} in the negative.
We then go a step further and present a $q$-matroid that is not representable by any $\F_q$-linear rank-metric code.
It remains an open question whether representability of a $q$-matroid by an $\F_q$-linear code implies representability by an
$\F_{q^m}$-linear code.

It is known from~\cite{GJLR19} that if $m\geq n$, then the column $q$-polymatroid of an MRD code is the uniform $q$-matroid of
rank $n-d+1$, where~$d$ is the rank distance of the code.
It is easily seen that this is not the case for $m<n$.
However, it turns out that for $m=n-1$ the $q$-polymatroid is fully determined by the parameters $(m,d,q)$ of the MRD code (and is not a $q$-matroid unless $d=1$), while this is not the case anymore if $m<n-1$.

Finally, we introduce a closure operator and the corresponding notion of flats for $q$-polymatroids.
They generalize the analogous notions for matroids and $q$-matroids.
However, just like for polymatroids, the lattice of flats of a $q$-polymatroid does not enjoy the same properties as for 
$q$-matroids.
Nonetheless, if the $q$-polymatroid arises from a rank-metric code, the collection of flats turns out to be closely related to the code:
the generalized weights of the code are fully determined by the flats.
In this context, we will also clarify a subtle issue between the various definitions of generalized weights for square $\F_{q^m}$-linear rank-metric codes.

\medskip
\textbf{Notation:}
Throughout, let $\F=\F_q$ be a field of order~$q$ and~$E$ be a finite-dimensional $\F$-vector space.
We write $V\leq W$ if~$V$ is a subspace of the vector space~$W$.
A $k$-dimensional subspace is called a $k$-space.
The collection of subspaces of a vector space~$X$ is denoted by~$\cV(X)$.
The notation $[n]$ is used for the set $\{1,\ldots,n\}$.
Finally, $e_i,i\in[\ell],$ and $E_{ij},i\in[\ell],j\in[m],$ denote the standard basis vectors in~$\F^\ell$ and the standard basis matrices in
$\F^{\ell\times m}$, respectively (that is, $E_{ij}$ has entry~$1$ at position~$(i,j)$ and~$0$ elsewhere).

\section{Preliminaries on $q$-Polymatroids}\label{S-Prelims}
The following definition is from \cite[Def.~4.1]{GJLR19} of Gorla et al., with the sole difference that we require rank functions to assume
rational values.

\begin{defi}\label{D-PMatroid}
Set $\cV=\cV(E)$.
A \emph{$q$-rank function} on~$E$ is a map  $\rho: \cV\longrightarrow\Q_{\geq0}$  satisfying:
\begin{mylist2}
\item[(R1)\hfill] Dimension-Boundedness: $0\leq\rho(V)\leq \dim V$  for all $V\in\cV$;
\item[(R2)\hfill] Monotonicity: $V\leq W\Longrightarrow \rho(V)\leq \rho(W)$  for all $V,W\in\cV$;
\item[(R3)\hfill] Submodularity: $\rho(V+W)+\rho(V\cap W)\leq \rho(V)+\rho(W)$ for all $V,W\in\cV$.
\end{mylist2}
A \emph{$q$-polymatroid (\qPM) on~$E$} is a pair $(E,\rho)$, where $\rho: \cV\longrightarrow\Q_{\geq0}$ is a
$q$-rank function.
The value $\rho(E)$ is called the \emph{rank} of the \qPM.
If $\rho$ is the zero map, we call the \qPM{} \emph{trivial}.
\end{defi}

The following additional notions for \qPM{}s will be useful.

\begin{defi}\label{D-PMatroidNotions}
Let $\cM=(E,\rho)$ be a \qPM{}.
A number $\mu\in\Q_{>0}$ is a \emph{denominator} of~$\rho$ (and $\cM$) if $\mu\rho(V)\in\N_0$ for all $V\in\cV(E)$.
The smallest denominator is called the \emph{principal denominator}.
A \qPM{} with principal denominator~$1$ (i.e., $\rho(V)\in\N_0$ for all~$V$) is called a \emph{$q$-matroid}.
We declare~$1$ the principal denominator of the trivial \qPM{}.
\end{defi}

Let us relate our definition  to the literature.
First of all, a $q$-matroid in the sense of Jurrius/ Pellikaan~\cite{JuPe18} is exactly a $q$-matroid as defined above.
Next, as already mentioned, our definition coincides with that in \cite[Def.~4.1]{GJLR19} by Gorla et al.\ except that our rank functions take rational values.
As we will see in Section~\ref{S-RMCMatroids}, this is indeed the case for the \qPM{}s induced by rank-metric codes.
Finally, for any $r\in\N$ a $(q,r)$-polymatroid as in \cite[Def.~2]{Shi19} by Shiromoto can be turned into a \qPM{} with denominator~$r$ by
dividing the rank function by~$r$.
Conversely, given a \qPM{} $(E,\rho)$ with denominator~$\mu$, then $(E,\mu\rho)$ is a $(q,\lceil\mu\rceil)$-polymatroid in the sense of \cite{Shi19}.
Finally, denominators arise for the definition of independent spaces of \qPM{}s; see \cite{GLJ21I}.

\begin{rem}\label{R-PropertiesPMatroid}
\begin{arabiclist}
\item Every denominator~$\mu$ of a \qPM{} $(E,\rho)$ satisfies $\mu\geq1$.
         Indeed, by (R1) $\rho(V)\leq 1$ for all $1$-spaces~$V$, and by~(R3) $\rho$ is the zero map
         if and only if $\rho(V)=0$ for all $1$-spaces~$V$.
\item Let $(E,\rho)$ be a non-trivial \qPM. For $V\in\cV(E)$ write $\rho(V)=\alpha_V/\beta_V$ with
         $\alpha_V,\,\beta_V\in\N$ relatively prime.
         Then the principal denominator is given by
         \[
             \mu=\frac{\text{lcm}\{\beta_V\mid V\in\cV(E)\}}{\gcd\{\alpha_V\mid V\in\cV(E)\}},
          \]
          and $\mu\N$ is the set of all denominators of $(E,\rho)$.
\end{arabiclist}
\end{rem}

The following $q$-matroids will occur occasionally. They can also be found at \cite[Ex.~4]{JuPe18}.
One easily verifies that the map~$\rho$ is indeed a rank function.

\begin{defi}\label{D-Uniform}
Let $\dim E=\ell$.
Fix  $k\in[\ell]$ and define $\rho(V)=\min\{k,\dim V\}$ for $V\in\cV(E)$. Then $(E,\rho)$ is a $q$-matroid.
It is called the \emph{uniform matroid on~$E$ of rank~$k$} and denoted by $\cU_{k}(E)$.
\end{defi}

Some of the basic properties for $q$-matroids derived in \cite[Sec.~3]{JuPe18} hold true for \qPM{}s as well.
We spell out the following ones, which we will need later on.
The proofs are identical to the ones in \cite[Prop.~6 and~7]{JuPe18}.

\begin{prop}\label{P-RankVx}
Let $(E,\rho)$ be a \qPM.
\begin{alphalist}
\item Let $V,W\in\cV(E)$.  Suppose $\rho(V+\subspace{x})=\rho(V)$ for all $x\in W$.
         Then $\rho(V+W)=\rho(V)$.
\item Let $V\in\cV(E)$ and $X,Y\in\cV(E)$ be $1$-spaces such that $\rho(V)=\rho(V+X)=\rho(V+Y)$.
        Then $\rho(V+X+Y)=\rho(V)$.
\end{alphalist}
\end{prop}

The following notion of equivalence is from  \cite[Def.~4.4]{GJLR19}.

\begin{defi}\label{D-EquivMatroid}
Two \qPM{}s $\cM_i=(E_i,\rho_i),\,i=1,2,$  are \emph{equivalent}, denoted by $\cM_1\approx\cM_2$,
if there exists an $\F$-isomorphism $\alpha\in\Hom_{\F}(E_1,E_2)$ such that $\rho_2(\alpha(V))=\rho_1(V)$ for all $V\in\cV(E_1)$.
\end{defi}

At this point we want to briefly discuss a more general notion of equivalence for \qPM{}s.

\begin{rem}\label{R-TrafoExact}
Two \qPM{}s $\cM_i=(E_i,\rho_i),\,i=1,2,$  are \emph{scaling-equivalent}
if there exists an $\F$-isomorphism $\alpha\in\Hom_{\F}(E_1,E_2)$ and $a\in\Q_{>0}$ such that $\rho_2(\alpha(V))=a\rho_1(V)$ for all $V\in\cV(E_1)$.
This notion makes sense for \qPM{}s because there exist non-trivial \qPM{}s that do not attain the upper bound in (R1) non-trivially.
We briefly elaborate.
Let us call a \qPM{} $\cM=(E,\rho)$ \emph{exact} if  there exists some nonzero space $\hat{V}\in\cV$ such that $\rho(\hat{V})=\dim\hat{V}$.
Clearly, a non-trivial $q$-matroid is exact, but there exist non-trivial non-exact \qPM{}s; see \cref{E-NotNiceMatroids} in the next section.
It follows immediately from the submodularity in (R3) that a \qPM{}~$\cM=(E,\rho)$ is exact if and only if there exists a $1$-space~$V$ such that $\rho(V)=1$.
This implies that any denominator of an exact \qPM{} is an integer.
One can turn a non-exact \qPM{} $(E,\rho)$  into an exact one using scaling-equivalence.
Indeed, suppose $\rho(V)<\dim V$ for all $V\in\cV(E)\setminus0$.
Let $a=\max\{\rho(V)/\dim V\mid V\in\cV(E)\setminus 0\}$.
Then $a\in\Q_{>0}$ and there exists $\hat{V}\in\cV(E)$ such that $a=\rho(\hat{V})/\dim\hat{V}$.
Thus $(E,a^{-1}\rho)$ is an exact \qPM.
\end{rem}

We close this section with introducing the dual \qPM.
It is a straightforward generalization of duality of matroids based on the rank function (see, e.g. \cite[Prop.~2.1.9]{Ox11}), but
requires more details when replacing set-theoretic complements by orthogonal spaces.
Since we define \qPM{}s over arbitrary ground spaces, we need to specify a non-degenerate symmetric bilinear form,
and, not surprisingly, the dual rank function depends on the choice of this form.
But as we will see, different forms lead to equivalent dual \qPM{}s.
This generality is needed in order to discuss deletions and contractions later on.
Part of the following result is from \cite[4.5--4.7]{GJLR19} (see also \cite[Thm.~42]{JuPe18} for $q$-matroids).

\begin{theo}\label{T-DualqPM}
Let $\inner{\cdot}{\cdot}$ be a non-degenerate symmetric bilinear form on~$E$.
For $V\in\cV(E)$ define $V^\perp=\{w\in E\mid \inner{v}{w}=0\text{ for all }v\in V\}$.
Let $\cM=(E,\rho)$ be a \qPM{} and set
\begin{equation}\label{e-rhodual}
    \rho^*(V)=\dim V+\rho(V^\perp)-\rho(E).
\end{equation}
Then $\rho^*$ is a $q$-rank function on~$E$ and $\cM^*=(E,\rho^*)$ is a \qPM. It is called the \emph{dual} of~$\cM$.
Furthermore, $\cM^{**}=\cM$, where $\cM^{**}=(\cM^*)^*$ is the bidual, and $\cM$ and $\cM^*$ have the same set of denominators.
Finally, the equivalence class of~$\cM^*$ does not depend on the choice of the non-degenerate symmetric bilinear form.
More precisely, if $\langle\!\inner{\cdot}{\cdot}\!\rangle$ is another non-degenerate symmetric bilinear form on~$E$ and
$\cM^{\hat{*}}=(E,\rho^{\hat{*}})$ is the resulting dual \qPM{}, then $\cM^{\hat{*}}\approx\cM^*$.
\end{theo}

\begin{proof}
The fact that $\rho^*$ is a $q$-rank function and the identity $\rho^{**}=\rho$ have been proven in \cite[Thms.~4.6, 4.7]{GJLR19}.
The statement about the denominators is obvious.
It remains to show the very last statement.
Thus, let $\langle\!\inner{\cdot}{\cdot}\!\rangle$ be another non-degenerate symmetric bilinear form on~$E$.
For $V\in\cV(E)$ denote by $V^\perp$ and $V^{\pperp}$ the orthogonal spaces of~$V$ with respect to $\inner{\cdot}{\cdot}$ and
$\langle\!\inner{\cdot}{\cdot}\!\rangle$, respectively.
Let $v_1,\ldots,v_\ell$ be a basis of~$E$ and $\psi:E\longrightarrow\F^\ell$ be the associated coordinate map.
Let $Q=(\inner{v_i}{v_j}),\,\hat{Q}=(\langle\!\inner{v_i}{v_j}\!\rangle)\in\F^{\ell\times\ell}$ be the Gram matrices associated to the bilinear forms.
Then $Q,\,\hat{Q}$ are symmetric and nonsingular and we have
$\inner{v}{w}=\psi(v)Q\psi(w)\T$ and $\langle\!\inner{v}{w}\!\rangle=\psi(v)\hat{Q}\psi(w)\T$ for all $v,w\in E$.
Define the automorphism
\begin{equation}\label{e-phi}
    \phi:E\longrightarrow E,\ v\longmapsto \psi^{-1}(\psi(v)\hat{Q}Q^{-1}).
\end{equation}
Now we have for any $V\in\cV(E)$ and $w\in E$
\begin{align*}
  w\in\phi(V)^\perp&\Longleftrightarrow \psi(\phi(v))Q\psi(w)\T=0\text{ for all }v\in V\\
       &\Longleftrightarrow \psi(v)\hat{Q}Q^{-1}Q\psi(w)\T=0\text{ for all }v\in V\\
       &\Longleftrightarrow w\in V^{\pperp}.
\end{align*}
Hence $V^{\pperp}=\phi(V)^\perp$ and thus $\rho^*(\phi(V))=\rho^{\hat{*}}(V)$ for all $V\in\cV(E)$.
This shows that $\cM^*$ and $\cM^{\hat{*}}$ are equivalent.
\end{proof}

The next result has been proven in~\cite{GJLR19} for \qPM{}s on $\F^\ell$, endowed with the standard dot product.
Thanks to the just proven invariance of the dual, it generalizes as follows without the need to specify bilinear forms.

\begin{prop}[\mbox{\cite[Prop.~4.7]{GJLR19}}]\label{P-EquivDual}
Let $\cM=(E,\rho)$ and $\hat{\cM}=(\hat{E},\hat{\rho})$ be \qPM{}s.
Then $\cM\approx\hat{\cM}$ implies $\cM^*\approx\hat{\cM}^*$.
\end{prop}

\begin{exa}[\mbox{\cite[Ex.~47]{JuPe18}}]\label{E-DualUnif}
It is easy to see that $\cU_{k}(E)^*=\cU_{\dim E-k}(E)$.
\end{exa}

In Section~\ref{S-DelContr} we will show that duality of \qPM{}s corresponds to duality of rank-metric
codes and is compatible with duality of puncturing and shortening of codes.
In \cite{GLJ21I} we  study independent spaces and bases and their behavior under dualization.

\section{Rank-Metric Codes and the Induced $q$-Polymatroids}\label{S-RMCMatroids}
In this section we study \qPM{}s associated to linear rank-metric codes as introduced in~\cite{GJLR19}.
We first collect some well-known facts for codes in $\Fnm$.
As usual, we endow $\Fnm$ with the rank-metric, defined as $\dd(A,B)=\rk(A-B)$.
Hence $\dd(A,0)=\rk(A)$ for all $A\in\Fnm$.

Throughout, a  \emph{rank-metric code} is meant to be linear, that is, it is a subspace of the metric space $(\Fnm,\dd)$.
Part~(a)--(c) of the following proposition is standard knowledge on rank-metric codes, see for instance~\cite{Gor19},
and Part~(d) can be found in \cite[Lem.~28]{Ra16a}.
For $V\leq\F^\ell$ denote by $V^\perp\leq\F^\ell$ the orthogonal space with respect to the standard dot product.

\begin{defiprop}\label{D-RMCBasics}
Let  $\cC\leq \Fnm$ be a rank-metric code.
\begin{alphalist}
\item The \emph{rank distance} of~$\cC$ is defined as $\drk(\cC)=\min\{\rk(M)\mid M\in\cC\setminus0\}$.
\item If $d=\drk(\cC)$, then $\dim(\cC)\leq \max\{m,n\}(\min\{m,n\}-d+1)$, which is known as the \emph{Singleton bound}.
        If  $\dim(\cC)=\max\{m,n\}(\min\{m,n\}-d+1)$, then $\cC$ is called an \emph{MRD code}.
\item The \emph{dual code} of~$\cC$ is defined as $\cC^\perp=\{M\in\Fnm\mid \tr(MN\T)=0\text{ for all }N\in\cC\}$, where
        $\tr(\cdot)$ denotes the trace of the given matrix.
       If $\cC$ is an MRD code with rank distance~$d$, then $\cC^\perp$ is an MRD code with rank distance $\min\{m,n\}-d+2$.
\item For $V\in\cV(\F^n)$ and $W\in\cV(\F^m)$ we define the following subspaces of~$\cC$, which are known as \emph{shortenings}:
\[
   \cC(V,\cc)=\{M\in\cC\mid \colsp(M)\leq V\}\ \text{ and }\
   \cC(W,\rr)=\{M\in\cC\mid \rowsp(M)\leq W\},
\]
where $\colsp(M)$ and $\rowsp(M)$ denote the column space and row space of~$M$, respectively.
Then $\Fnm(V,\cc)^\perp=\Fnm(V^\perp,\cc)$ and $ \Fnm(W,\rr)^\perp=\Fnm(W^\perp,\rr)$ and
\begin{align}
  \dim\cC(V^\perp,\cc)&=\dim\cC-m\dim V+\dim\cC^\perp(V,\cc), \label{e-DimCV}\\
  \dim\cC(W^\perp,\rr)&=\dim\cC-n\dim W+\dim\cC^\perp(W,\rr).\nonumber
\end{align}
\end{alphalist}
\end{defiprop}

Now we are ready to introduce \qPM{}s associated to a rank-metric code.
The following definition and the first statement are from~\cite{GJLR19}.
The statements in \eqref{e-rhoV} are immediate consequences of Proposition~\ref{D-RMCBasics}(d).

\begin{prop}[\mbox{\cite[Thm.~5.3]{GJLR19}}]\label{P-RMCMatroid}
Let $\cC\leq\F^{n\times m}$ be a nonzero rank-metric code.  Define the maps
\begin{align*}
   \rho_\cc:&\cV(\F^n)\longrightarrow \Q_{\geq0},\quad V\longmapsto \frac{\dim\cC-\dim\cC(V^\perp,\cc)}{m},\\[1ex]
   \rho_\rr:&\cV(\F^m)\longrightarrow \Q_{\geq0},\quad W\longmapsto \frac{\dim\cC-\dim\cC(W^\perp,\rr)}{n}.
\end{align*}
Then  $\rho_\cc$ and $\rho_\rr$ are  $q$-rank functions with denominators~$m$ and~$n$, respectively.
Furthermore,
\begin{equation}\label{e-rhoV}
   \rho_\cc(V)=\dim V -\frac{1}{m}\dim\cC^\perp(V,\cc)
     \ \text{ and }\
   \rho_\rr(W)=\dim W-\frac{1}{n}\dim\cC^\perp(W,\rr).
\end{equation}
\end{prop}

The denominators~$m$ and~$n$ are in general not principal.
Note that in the notation we suppress the dependence of these maps on the code~$\cC$.

\begin{defi}\label{D-RMCMatroid}
Let $\cC\leq\F^{n\times m}$ be a nonzero rank-metric code.
The \qPM{}s  $\cM_\cc(\cC):=(\F^n,\rho_\cc)$ and $\cM_\rr(\cC):=(\F^m,\rho_\rr)$ are called the
\emph{column $q$-polymatroid} and \emph{row $q$-polymatroid} of~$\cC$, respectively.
Their ranks are $\dim\cC/m$ and $\dim\cC/n$, respectively.
\end{defi}

The expressions in~\eqref{e-rhoV} show immediately that if $ \cC_1\leq\cC_2$, then $\rho_{\cc,1}(V)\leq \rho_{\cc,2}(V)$
for all $V\in\cV(\F^n)$,  where $\rho_{\cc,i}$ is the column rank function of~$\cC_i$. Similarly for the row \qPM.
This has also been proven in \cite[Lem.~22]{GhJo20}.

Equivalence of codes, in the following (standard) sense, translates into equivalence of the associated \qPM{}s.

\begin{defi}\label{D-EquivCodes}
Let $\cC,\,\cC'\leq\Fnm$ be rank-metric codes. We call $\cC$ and $\cC'$ \emph{equivalent} if there exist
matrices $X\in\GL_n(\F),\,Y\in\GL_m(\F)$ such that $\cC'=X\cC Y:=\{XMY\mid M\in\cC\}$.
If $n=m$, we call~$\cC,\,\cC'$ \emph{transposition-equivalent} if there exist
matrices $X,Y\in\GL_n(\F)$ such that $\cC'=X\cC\T Y:=\{XM\T Y\mid M\in\cC\}$.
\end{defi}

The proof of the next result is straightforward with~\eqref{e-rhoV} by noting that $\cC'=X\cC Y$ implies
$(\cC')^\perp=(X^{-1})\T\cC^\perp(Y^{-1})\T$ and $\cC'=X\cC\T Y$ implies $(\cC')^\perp=(X^{-1})\T(\cC^\perp)\T(Y^{-1})\T$.
For an alternative proof see \cite[Prop.~6.7]{GJLR19}.

\begin{prop}\label{P-EquivCodeMatroid}
Let  $\cC,\,\cC'\leq\Fnm$ be rank-metric codes.
\begin{alphalist}
\item Suppose $\cC,\,\cC'$ are equivalent, say $\cC'=X\cC Y$ for some $X\in\GL_n(\F),\,Y\in\GL_m(\F)$.
        Then $\cM_\cc(\cC)$ and $\cM_\cc(\cC')$ are equivalent via $\beta\in\Hom_{\F}(\F^n,\F^n)$
        given by $x\mapsto (X\T)^{-1}x$.
        Similarly, $\cM_\rr(\cC)$ and $\cM_\rr(\cC')$ are equivalent via the isomorphism $\alpha\in\Hom_{\F}(\F^m,\F^m)$
        given by $x\mapsto x(Y\T)^{-1}$.
\item Let $n=m$ and suppose $\cC,\,\cC'$ are transposition-equivalent, say  $\cC'=X\cC\T Y$ for $X,Y\in\GL_n(\F)$.
       Then $\cM_\cc(\cC)$ and $\cM_\rr(\cC')$ are equivalent via~$\alpha$ and
       $\cM_\rr(\cC)$ and $\cM_\cc(\cC')$ are equivalent via~$\beta$ with $\alpha,\beta$ as in~(a).
\end{alphalist}
\end{prop}

Equivalence allows us to easily introduce $\F_{q^m}$-linear rank-metric codes.
Recall that $\F=\F_q$.
Let~$\omega$ be a primitive element of the field extension~$\F_{q^m}$ and $f=x^m-\sum_{i=0}^{m-1}f_ix^i\in\F[x]$ be its minimal polynomial over~$\F$.
We define the companion matrix
\begin{equation}\label{e-Deltaf}
    \Delta_f=
    \begin{pmatrix} &1& & \\ & &\ddots & \\ & & & 1\\ f_0&f_1&\cdots&f_{m-1}\end{pmatrix}\in\GL_m(\F).
\end{equation}
Let $\psi:\F_{q^m}\longrightarrow\F^m$ be the coordinate map with respect to the basis $(1,\omega,\ldots,\omega^{m-1})$.
Extending this map entry-wise, we obtain, for any~$n$, an isomorphism
\begin{equation}\label{e-CoordMap}
      \Psi:\F_{q^m}^n\longrightarrow\F^{n\times m},\
      \begin{pmatrix}x_1& \cdots& x_n\end{pmatrix}\longmapsto \begin{pmatrix}\psi(x_1)\\ \vdots\\ \psi(x_n)\end{pmatrix}.
\end{equation}
Now multiplication of  $c:=\sum_{i=0}^{m-1}c_i\omega^i\in\F_{q^m}$ by~$\omega$ corresponds to
right multiplication of its coordinate vector $\psi(c)=(c_0,\ldots,c_{m-1})\in\F^m$ by $\Delta_f$.
Therefore, an $\F$-linear subspace~$\cC$ of $\F_{q^m}^n$ is $\F_{q^m}$-linear if and only if its image
$\Psi(\cC)\leq\F^{n\times m}$
is invariant under right multiplication by~$\Delta_f$.
Recall further that $\F[\Delta_f]$ is a subfield of order~$q^m$ of $\F^{m\times m}$, and more generally,
if $s$ is a divisor of~$m$, then  $\F[\Delta_f^{(q^m-1)/(q^s-1)}]$ is a subfield of order~$q^s$.
Allowing different bases for the coordinate map, we arrive at the following definition.

\begin{defi}\label{D-Fqn-linear}
Let $\cC\leq\F^{n\times m}$ be a rank-metric code (hence an $\F$-linear subspace).
Let $s$ be a divisor of~$m$ and set $M=(q^m-1)/(q^s-1)$.
Then $\cC$ is \emph{right $\F_{q^s}$-linear} if there exists an $X\in\GL_m(\F)$ such that
the code $\cC X$ is invariant under right multiplication by~$\Delta_f^M$.
\emph{Left linearity} over $\F_{q^n}$ and its subfields is defined analogously.
\end{defi}

Clearly, the qualifiers right/left are needed only in the case where $\F_{q^s}$ is a subfield of both $\F_{q^m}$ and $\F_{q^n}$.
Note that for $\tilde{\cC}:=\cC X$ we have $\tilde{\cC}(V,\cc)=\cC(V,\cc)X$ for all $V\in\cV(\F^n)$.
Furthermore, if~$\tilde{\cC}$  is invariant under right multiplication by~$\Delta_f^M$, then so is $\tilde{\cC}(V,\cc)$.
Hence $\tilde{\cC}(V,\cc)$ is right $\F_{q^s}$-linear (see also \cite[Lem.~19]{JuPe18} for the case $s=m$),
and thus its dimension over~$\F$ is a multiple of~$s$.
All of this leads to the following remark.

\begin{rem}\label{R-rhoFqnlinear}
Let $\cC\leq\F^{n\times m}$  be a right $\F_{q^s}$-linear rank-metric code for some subfield~$\F_{q^s}$ of $\F_{q^m}$.
Then $\mu=m/s$ is a denominator of the column \qPM{} $\cM_\cc(\cC)$.
In particular, for $s=m$ the \qPM{} $\cM_\cc(\cC)$ is a $q$-matroid.
These are exactly the $q$-matroids studied in \cite{JuPe18}.
Of course, the column \qPM{} of a code $\cC\leq\F^{n\times m}$ may be a $q$-matroid
even if~$\cC$ is not right $\F_{q^m}$-linear.
This is for instance the case for MRD codes in $\F^{n\times m}$ if $m\geq n$ (see Prop\-osition~\ref{P-MRDColMatroid}).
Analogous statements hold for the row \qPM.
\end{rem}

Let us return to general rank-metric codes in $\Fnm$.
From now on we will focus on the associated column \qPM{}s.
The discussion of the row \qPM{}s is analogous.
On several occasions we will have to pay close attention to the cases $n\leq m$ versus $n>m$.
We first record the following simple fact, which is immediate with Proposition~\ref{P-RMCMatroid}
(see also \cite[Prop.~6.2]{GJLR19} and \cite[Lem.~30]{JuPe18}).
Let $\cC\leq\F^{n\times m}$ be a nonzero rank-metric code with rank distance~$d$ and let~$d^\perp$ be the rank distance of~$\cC^\perp$.
Then for any $V\in\cV(\F^n)$
\begin{equation}\label{e-Rhocc}
   \rho_\cc(V)=\begin{cases} {\DS\frac{\dim\cC}{m}}, &\text{if }\dim V>n-d,\\[1ex] \dim V,&\text{if }\dim V<d^\perp.\end{cases}
\end{equation}

Together with \cref{D-RMCBasics}(b),(c) this immediately leads to the following result for MRD codes if $n\leq m$.
Note the similarity to the well-known correspondence between MDS block codes and uniform matroids.

\begin{prop}[\mbox{\cite[Cor.~6.6]{GJLR19}}]\label{P-MRDColMatroid}
Let $n\leq m$ and $\cC\leq\Fnm$ be an MRD code with $\drk(\cC)=d$.
Let~$\rho_\cc$ be the rank function of the associated column \qPM{} $\cM_\cc(\cC)$. Then
\[
    \rho_\cc(V)=\min\{n-d+1,\dim V\} \ \text{ for all } V\in\cV(\F^n).
\]
Hence $\cM_\cc(\cC)$ is the uniform matroid $\cU_{n-d+1}(\F^n)$.
In particular, $\cM_\cc(\cC)$ is a $q$-matroid, i.e., its principal denominator is~$1$.
\end{prop}

In order to discuss the \qPM{} associated to MRD codes for $m\leq n$ we need the following lemma.

\begin{lemma}\label{L-DimCV}
Let  $\cC\leq\Fnm$ be a rank metric code with $\drk(\cC)=d$.
Let $V\in\cV(\F^n)$ and $\dim V=v$. Then $\dim\cC(V,\cc)\leq\max\{m,v\}(\min\{m,v\}-d+1)$.
\end{lemma}

\begin{proof}
First assume that $V=\colsp(I_v\mmid 0)\T=\subspace{e_1,\ldots,e_v}$, where $e_i\in\F^n$ is the $i$-th standard basis vector.
Then we have a rank-preserving, thus injective, linear map
\[
    \pi:\cC(V,\cc)\longrightarrow\F^{v\times m},\ \begin{pmatrix}M\\0\end{pmatrix}\longmapsto M.
\]
Hence $\text{im}(\pi)$ is a rank-metric code in $\F^{v\times m}$ of rank distance at least~$d$, and
the upper bound for $\dim\cC(V,\cc)$ follows from the Singleton bound.
For the general case where $V=\subspace{x_1,\ldots,x_v}$, choose $Y\in\GL_n(\F)$ such that $Yx_i=e_i$ for $i\in[v]$ and
set $\cC'=Y\cC $.
Then $\cC'$ has rank distance~$d$ as well and $Y\cC(V,\cc)=\cC'(YV,\cc)$.
Since $YV=\subspace{e_1,\ldots,e_v}$, the upper bound on $\dim\cC(V,\cc)$ follows from the first part of this proof.
\end{proof}

Now we obtain the following information about the column \qPM{} of an MRD code if $m\leq n$.

\begin{theo}\label{T-MRDRowMatroid}
Let $m\leq n$ and $\cC\leq\Fnm$ be an MRD code with $\drk(\cC)=d$.
Furthermore, denote the dimension of $V\in\cV(\F^n)$ by~$v$.
Then the rank function~$\rho_\cc$ of the associated column \qPM{} $\cM_\cc(\cC)$ satisfies
\[
  \rho_\cc(V)=\left\{\begin{array}{cl}  v,&\text{if }v\leq m-d+1,\\[.7ex] \frac{n(m-d+1)}{m},&\text{if }v\geq n-d+1.\end{array}\right.
\]
Furthermore $\rho_\cc(V)\geq\max\{1,v/m\}(m-d+1)$ if $v\in[m-d+2,n-d]$.
\end{theo}

\begin{proof}
For $v\geq n-d+1$ we have $\rho_\cc(V)=m^{-1}\dim\cC=nm^{-1}(m-d+1)$ thanks to \eqref{e-Rhocc}.
Next, using~\eqref{e-rhoV} and the fact that $\drk(\cC^\perp)=m-d+2$ we arrive immediately at $\rho_\cc(V)=v$ if $v\leq m-d+1$.
The remaining statement follows from~\eqref{e-rhoV} along with Lemma~\ref{L-DimCV}, which applied to~$\cC^\perp$ yields
$\dim\cC^\perp(V,\cc)\leq m(v-m+d-1)$ if $v\leq m$ and $\dim\cC^\perp(V,\cc)\leq v(d-1)$ if $v\geq m$.
\end{proof}

If $m=n-1$, the interval $[m-d+2,n-d]$ is empty for any~$d$. Thus we have

\begin{cor}\label{C-MRDRowMatroid}
Let $\cC\leq\F^{n\times(n-1)}$ be an MRD code. Then $\cM_\cc(\cC)$ is fully determined by the parameters
$(n,\drk(\cC),|\F|)$ and is not a $q$-matroid (unless $\drk(\cC)=1$).
\end{cor}

\begin{rem}\label{R-RowMRDDenom}
In the situation of \cref{T-MRDRowMatroid} the denominator~$m$ is not necessarily principal (even if the code is not linear over a subfield of $\F_{q^m}$).
This is for instance the case for a $[7\times6;4]$-MRD code: $n(m-d+1)/m=7/2$, and thus the principal denominator is~$2$.
\end{rem}

If $m< n-1$, the column \qPM{} of an MRD code~$\cC$ in $\Fnm$ is not fully determined by the parameters $(n,m,\drk(\cC),|\F|)$.
This is illustrated by the following example.

\begin{exa}\label{E-RowMatMRD}
In $\F_2^{5\times 2}$ consider the codes~$\cC_1=\subspace{A_1,\dots,A_5}$ and $\cC_2=\subspace{B_1,\ldots,B_5}$, where
\begin{align*}
  &A_1=\begin{pmatrix}1&1\\1&0\\0&0\\1&0\\0&0\end{pmatrix},\
    A_2=\begin{pmatrix}1&1\\1&1\\1&0\\0&1\\0&0\end{pmatrix},\
    A_3=\begin{pmatrix}0&0\\0&0\\1&1\\0&0\\0&1\end{pmatrix},\
    A_4=\begin{pmatrix}0&0\\0&1\\0&0\\0&0\\1&1\end{pmatrix},\
    A_5=\begin{pmatrix}1&0\\0&1\\1&1\\0&0\\0&1\end{pmatrix},\\[3ex]
  & B_1=\begin{pmatrix}1&0\\0&1\\0&0\\0&0\\0&0\end{pmatrix},\
     B_2=\begin{pmatrix}0&0\\1&0\\0&1\\0&0\\0&0\end{pmatrix},\
     B_3=\begin{pmatrix}0&0\\0&0\\1&0\\0&1\\0&0\end{pmatrix},\
     B_4=\begin{pmatrix}0&0\\0&0\\0&0\\1&0\\0&1\end{pmatrix},\
     B_5=\begin{pmatrix}0&1\\0&0\\0&1\\0&0\\1&0\end{pmatrix}.
\end{align*}
Both codes are MRD with rank distance~$d=2$; in fact,~$\cC_2$ is a ($\F_{2^5}$-linear) Gabidulin code.
The \qPM{}s $\cM_\cc(\cC_1)$ and $\cM_\cc(\cC_2)$ are not equivalent.
Indeed, denote the two rank functions by $\rho_\cc^i,\,i=1,2,$ and consider $\rho_\cc^i(V)$ for
$\dim V\in\{2,3\}=[m-d+2,n-d]$.
For the $155$ subspaces of $\F_2^5$ of dimension~$2$ the map $\rho_\cc^1$ assumes the value~$1$ exactly once and the values
$3/2$ and $2$ exactly 28 and 126 times, respectively.
On the other hand, $\rho_\cc^2$ assumes the values~$3/2$ and~$2$ exactly 31 and 124 times, respectively, and never takes the value~$1$.
\end{exa}

We briefly return to the notion of scaling equivalence and non-exactness from \cref{R-TrafoExact}.

\begin{exa}\label{E-NotNiceMatroids}
For \qPM{}s induced by rank-metric codes non-exactness  is a very restrictive property.
Indeed, the first identity in \eqref{e-rhoV} shows that
\begin{equation}\label{e-nonexact}
   \cM_\cc(\cC)\text{ is not exact }\Longleftrightarrow \cC^\perp(V,\cc)\neq0\text{ for all }V\in\cV(\F^n)\setminus0.
\end{equation}
In particular, non-exactness of $\cM_\cc(\cC)$ implies $\drk(\cC^\perp)=1$.
An example of a code satisfying~\eqref{e-nonexact} is for instance
$\cC=\subspace{E_{13},\,E_{23},\,E_{33},\,E_{12}+E_{22}+E_{41}+E_{43}}\leq\F_2^{3\times4}$
(recall the notation from \cref{S-Intro}) with dual code
$\cC^\perp=
\subspace{E_{11},\,E_{21},\,E_{31},\,E_{41}+E_{43},\,E_{12}+E_{43},\,E_{22}+E_{43},\,E_{32},E_{42}}$.
The column \qPM{} $\cM_\cc(\cC)$ can be rescaled with factor $3/2$, leading to an exact \qPM{}~$\cM'$.
It turns out that $\cM'=\cM_{\cc}(\cC')$ for the code $\cC'\leq\F_2^{4\times 4}$ given by
$\cC'=\langle E_{11},E_{12},E_{21},E_{22},E_{31},E_{32},E_{13}+E_{23}+E_{44}$,
$E_{14}+E_{24}+E_{43}+E_{44}\rangle$.
It is not clear to us whether representability of \qPM{}s by rank-metric codes (see the next section) is preserved under rescaling.
\end{exa}

We close this section with the following important result showing that duality of \qPM{}s corresponds to trace-duality of codes.
Recall duality of \qPM{}s from \cref{T-DualqPM}.

\begin{theo}[\mbox{\cite[Thm.~8.1]{GJLR19}}]\label{T-TraceDual}
Let $\cC\leq\F^{n\times m}$ be a rank-metric code and $\cC^\perp\leq\F^{n\times m}$ be dual.
Then $\cM_\cc(\cC)^*=\cM_\cc(\cC^\perp)$, where $\cM_\cc(\cC)^*$ is the dual of $\cM_\cc(\cC)$ w.r.t.\ the standard dot product on~$\F^n$.
Analogously, $\cM_\rr(\cC)^*=\cM_\rr(\cC^\perp)$.
\end{theo}

\section{Representability of $q$-Polymatroids}\label{S-Repr}
In this section we discuss representability of \qPM{}s via rank-metric codes.
We will present various examples of $q$-matroids that are not representable via $\F_{q^m}$-linear rank-metric codes (thereby answering a question  from \cite[Sec.~11]{JuPe18}).
A smaller example of such a non-representable $q$-matroid appeared in \cite{CeJu21}, which came out after the first version
of the paper at hand.
Later in this section we will show that that $q$-matroid is not even representable via any $\F_q$-linear rank-metric code.
We first fix the following notions of representability.
Recall $\F_{q^m}$-linearity from \cref{D-Fqn-linear}.
As before, $\F=\F_q$.

\begin{defi}\label{D-Repr}
Let $E$ be an $n$-dimensional $\F$-vector space and $\cM=(E,\rho)$ be a \qPM.
\begin{alphalist}
\item $\cM$ is said to be $\Fnm$-\emph{representable} if there exists a rank-metric code $\cC\leq\Fnm$ such that
        $\cM\approx\cM_\cc(\cC)$.
\item Suppose $\cM$ is a $q$-matroid.
        Then $\cM$ is said to be $\F_{q^m}$-\emph{representable}  if there exists a right $\F_{q^m}$-linear code~$\cC\leq\Fnm$
        such that $\cM\approx\cM_\cc(\cC)$.
\end{alphalist}
\end{defi}

We start with $\F_{q^m}$-representability of $q$-matroids.
It will be necessary to consider $\F_{q^m}$-linear codes as subspaces of $\F_{q^m}^n$, as is indeed common.
Consider the isomorphism~$\Psi:\F_{q^m}^{n}\longrightarrow \F^{n\times m}$ from~\eqref{e-CoordMap}.
Then with any $\F_{q^m}$-subspace~$\cC$ of $\F_{q^m}^{n}$ we can associate the column \qPM{} $\cM_\cc(\Psi(\cC))$,
which is in fact a $q$-matroid.
We denote this $q$-matroid by~$\cM_G$, where $G\in\F_{q^m}^{k\times n}$ is a generator matrix of~$\cC$
(that is, its rows form a basis of~$\cC$).
This is  the approach taken in \cite{JuPe18}.
The following lemma has been proven in \cite[Sec.~5]{JuPe18}.
It determines the rank function of $\cM_G$ with the aid of~$G$.
For self-containment we include a short proof using our notation.

\begin{lemma}\label{L-GYT}
Let $\cC\leq\F_{q^m}^n$ be an $\F_{q^m}$-linear rank-metric code with generator matrix~$G\in\F_{q^m}^{k\times n}$, where
$k=\dim_{\F_{q^m}}\cC$.
Consider the associated $q$-matroid $\cM_G=(\F^n,\rho_\cc)$.
Let $V\in\cV(\F^n)$ with $\dim V=t$, and let $Y\in\F^{n\times t}$ be such that $V=\colsp(Y)$.
Then
\[
   \rho_\cc(V)=\rk_{\F_{q^m}}(G Y).
\]
\end{lemma}

\begin{proof}
Clearly $\rk_{\F_{q^m}}(G Y)$ does not depend on the choice of~$Y$.
Since $Y$ has entries in~$\F$, any $x\in\F_{q^m}^n$ satisfies
$xY=0\Longleftrightarrow Y\T x\T=0\Longleftrightarrow Y\T\Psi(x)=0\Longleftrightarrow\colsp(\Psi(x))\leq V^\perp$.
Set $\tilde{\cC}=\Psi(\cC)\leq\F^{n\times m}$.
Then the space $\tilde{\cC}(V^\perp,\cc)$  satisfies
$\Psi^{-1}(\tilde{\cC}(V^\perp,\cc))=\{x\in\cC\mid xY=0\}$.
Let $\pi_Y:\cC\longrightarrow \F_{q^m}^t$ be the $\F_{q^m}$-linear map given by $x\longmapsto x Y$.
Then $\ker\pi_Y= \Psi^{-1}(\tilde{\cC}(V^\perp,\cc))$  and $\text{im}\,\pi_Y=\rowsp_{\F_{q^m}}(GY)$.
Now the desired statement follows from
\[
   \rho_\cc(V)=\frac{\dim_{\F}\tilde{\cC}-\dim_{\F}\tilde{\cC}(V^\perp,\cc)}{m}= \
   \dim_{\F_{q^m}}\tilde{\cC}-\dim_{\F_{q^m}}\Psi^{-1}(\tilde{\cC}(V^\perp,\cc))=\rk_{\F_{q^m}}(GY).
   \qedhere
\]
\end{proof}

\begin{rem}\label{R-RhoGenMat}
The above lemma generalizes as follows.
Consider a general rank-metric code $\cC\leq\F^{n\times m}$ with $\F$-dimension~$k$.
Taking the pre-image under~$\Psi$ of a basis of~$\cC$, we obtain a matrix $G\in\F_{q^m}^{k\times n}$ such that
$\Psi^{-1}(\cC)=\rowsp_{\F}(G)$, where the latter is defined as the $\F$-subspace of~$\F_{q^m}^n$
generated by the rows of~$G$.
Denote its $\F$-dimension by $\rowrk(G)$.
Then it is  easy to see that for $V$ as in \cref{L-GYT} we have
$\rho_\cc(V)=\rowrk(GY)/m$.
\end{rem}

We return to $\F_{q^m}$-representability of $q$-matroids. Of course,  a $q$-matroid $\cM=(\F^n,\rho)$ may be representable over some field $\F_{q^m}$, but not over any smaller field extension of $\F$. This is for instance the case for $\cM=\cM_G$, where
\[
   G=\begin{pmatrix}1&0&0&\alpha \\ 0&1&\alpha^2&\alpha^4\end{pmatrix}\in\F_{2^4}^{2\times 4},
\]
where $\alpha^4+\alpha+1=0$.
It can be verified with SageMath or any other computer algebra system  that $\cM_G=(\F_2^4,\rho_\cc)$ is not representable over $\F_{2^l}$ for $l\leq3$.
In fact,~$\cM_G$ is not even $\F_2^{4\times l}$-representable for any $l\leq3$.

We now construct $q$-matroids that are not $\F_{q^m}$-representable for any $m\in\N$.
In order to do so, we will make use of non-representable (classical) matroids.
Recall that a matroid is a pair $M=(X,r)$, where $X$ is a finite set and $r:2^{X}\longrightarrow\N_0$ satisfies (R1)--(R3) from \cref{D-PMatroid} if we replace the dimension, resp.\ sum, of subspaces by the cardinality, resp.\ union, of subsets.
The rank of~$M$ is defined as $r(X)$.
A matroid $M=(X,r)$ is called \emph{representable} over the field~$F$ if there exists $k\in\N$ and a matrix $G\in F^{k\times |X|}$
with columns indexed by the elements of~$X$,
such that for any subset $A\subseteq X$ we have $r(A)=\rk(G_A)$, where $G_A\in F^{k\times|A|}$ is the submatrix of~$G$ consisting of the columns with indices in~$A$.
It is easy to see that if such a matrix exists then we may choose $k=r(X)$.

The next result will provide us with a crucial link between $q$-matroids and matroids.

\begin{theo}\label{T-qPMClassMatr}
Let  $\cM=(\F^n,\rho)$ be a $q$-matroid and $\cB=\{v_1,\ldots,v_n\}$ be a basis of~$\F^n$.
Define
\[
        r:2^{\cB}\longrightarrow\N_0,\quad A\longmapsto\rho(\subspace{A}).
\]
Then $(\cB,r)$ is a  matroid, denoted  by $M(\cM,\cB)$.
Furthermore, if~$\cM$ is $\F_{q^m}$-representable, then $M(\cM,\cB)$ is $\F_{q^m}$-representable.
\end{theo}

\begin{proof}
1) We show that $r$ is indeed a rank function. (R1) and (R2) are clear.
For (R3) let $A,B\subseteq \cB$.
By basic Linear Algebra $\subspace{A\cup B}=\subspace{A}+\subspace{B}$ and $\subspace{A\cap B}=\subspace{A}\cap\subspace{B}$, and therefore
\begin{align*}
   r(A\cup B)+r(A\cap B)&=\rho(\subspace{A\cup B})+\rho(\subspace{A\cap B})=\rho(\subspace{A}+\subspace{B})+\rho(\subspace{A}\cap \subspace{B})\\
   &\leq \rho(\subspace{A})+\rho(\subspace{B})=r(A)+r(B).
\end{align*}
2) We now turn to the statement about representability. Let $\cM=\cM_G$ for some $G\in\F_{q^m}^{k\times n}$.
Define the matrix $X=(v_1,\ldots,v_n)$, which is in $\GL_n(\F)$, and set $G':=GX$.
Index the columns of~$G'$ by $v_1,\ldots,v_n$, and for any $A\subseteq\cB$
denote by $G'_{A}$ the submatrix of~$G'$ with columns indexed by $A$.
Let now $A=\{v_{a_1},\ldots,v_{a_l}\}$ be any subset of~$\cB$ and $Y\in\F^{n\times l}$ be the matrix with columns
$v_{a_1},\ldots,v_{a_l}$.
Then $\subspace{A}=\colsp(Y)$ and thus
\[
    r(A)=\rho(\subspace{A})=\rk(GY)=\rk(G'X^{-1}Y)=\rk(G'(e_{a_1},\ldots,e_{a_l}))=\rk(G'_A).
\]
This shows that the matroid $M(\cB,r)$ is represented by the matrix~$G'\in\F_{q^m}^{k\times n}$.
\end{proof}

Having established this particular relation between $q$-matroids and induced matroids, we now turn to a class of matroids
that can be interpreted as such induced matroids.
They form a special case of paving matroids.
A matroid is called \emph{paving} if no circuit has cardinality less than the rank of the matroid.
The following result is well known; see \cite[1.3.10]{Ox11} for a more general statement.
We include an elementary proof which  generalizes immediately to $q$-matroids.

\begin{prop}\label{P-MatroidExcSet}
Let $\cB$ be a set with $|\cB|=n$ and let $k\in[n]$.
Let $\cA$ be a collection of $k$-subsets of~$\cB$ such that $|A\cap B|\leq k-2$ for all
distinct $A,\,B\in\cA$.
Define the map
\[
  r:2^{\cB}\longrightarrow\N_0,\ X\longmapsto \begin{cases}\quad k-1,&\text{if }X\in\cA,\\ \min\{|X|,k\},&\text{if }X\not\in\cA.\end{cases}
\]
Then $r$ is a rank function on~$\cB$. We denote the resulting matroid by $M_{\cB,\cA}$.
The circuits are given by the subsets in~$\cA$ and all $(k+1)$-subsets that do not contain a subset in~$\cA$.
\end{prop}

\begin{proof}
(R1) and (R2) are clear.
For (R3) we have to show $r(A\cup B)\leq r(A)+r(B)-r(A\cap B)$ for all subsets $A,B\subseteq\cB$.
We may assume $A\not\subseteq B$ and $B\not\subseteq A$ for otherwise the statement is clear.
Note that $r(A\cup B)\leq k$ for all subsets.
Consider $S:=r(A)+r(B)-r(A\cap B)$ and note that $S\geq r(A)+r(B)-|A\cap B|$.
We proceed by cases.
\\
1) If $r(A)=r(B)=k$, then  $S=2k-r(A\cap B)\geq k$.
\\
2) If $r(A)=|A|$ and $r(B)=k$, then $S\geq |A|+k-|A\cap B|\geq |A|+k-(|A|-1)=k+1$.
\\
3) If $r(A)=|A|$ and $r(B)=|B|$, then $S\geq |A|+|B|-|A\cap B|=|A\cup B|\geq r(A\cup B)$.
\\
4) If $A,\,B\in\cA$, then $S\geq 2k-2-|A\cap B|\geq k$.
\\
5) If $A\in\cA$ and $r(B)=k$, then $S\geq 2k-1-|A\cap B|\geq 2k-1-(k-1)=k$.
\\
6) If $A\in\cA$ and $r(B)=|B|$, then $S\geq  k-1+|B|-|A\cap B|\geq k-1+|B|-(|B|-1)=k$.
\end{proof}

The $q$-analogue reads as follows.
The proof is entirely analogous to the previous  one: just replace cardinality and union of subsets by dimension
and sum of subspaces, respectively.

\begin{prop}\label{P-qMatroidExcSpace}
Let $n\in\N$ and fix an integer $k\in[n]$. Let $\cS$ be a collection of $k$-spaces in~$\F^n$ such that
$\dim(V\cap W)\leq k-2$ for all  distinct $V,\,W\in\cS$.
Define the map
\[
  \rho:\cV(\F^n)\longrightarrow\N_0,\ V\longmapsto \begin{cases}\quad k-1,&\text{if }V\in\cS,\\ \min\{\dim V,k\},&\text{if }V\not\in\cS.\end{cases}
\]
Then $(\F^n,\rho)$ is a $q$-matroid. We denote it by $\cM_{n,\F,\cS}$.
\end{prop}

We can easily link the two constructions.

\begin{prop}\label{P-LinkMatrqMatr}
Let $\cB$ be a basis of $\F^n$ and let $\cA$ be a collection of $k$-subsets of~$\cB$ such that $|A\cap B|\leq k-2$ for all
distinct $A,\,B\in\cA$.
Set $\cS=\{\subspace{A}\mid A\in\cA\}$. Then $\dim(V\cap W)\leq k-2$ for all $V,\,W\in\cS$ and
$M_{\cB,\cA}=M(\cM_{n,\F,\cS},\,\cB)$.
\end{prop}

\begin{proof}
It is easy to see that $\dim(V\cap W)\leq k-2$ for all $V,\,W\in\cS$.
For the second statement let $M_{\cB,\cA}=(\cB,r)$ and $\cM_{n,\F,\cS}=(\F^n,\rho)$.
Then we have for any subset $X\subseteq\cB$
\[
  r(X)=\Bigg\{\begin{array}{cl} k-1,&\text{if }X\in\cA,\\[.5ex] \min\{|X|,k\},&\text{if }X\not\in\cA,\end{array}\Bigg\}
  \ \text{ and }\
  \rho(\subspace{X})=\Bigg\{\begin{array}{cl} k-1,&\text{if }\subspace{X}\in\cS,\\[.5ex] \min\{\dim\subspace{X},k\},&\text{if }\subspace{X}\not\in\cS.\end{array}\Bigg\}
\]
This proves the desired statement.
\end{proof}

Now we are ready to present examples of $q$-matroids with ground space $\F^n$ that are not
$\F_{q^m}$-representable for any $m\in\N$.
In each case the non-representability follows from the non-representa\-bi\-li\-ty of the associated matroid with the aid of
\cref{T-qPMClassMatr}.
The following examples are universal in the sense that they apply to a large class of fields~$\F$.

\begin{exa}\label{E-NonRepr}
\begin{alphalist}
\item The Vamos Matroid~\cite[Ex.~2.1.25]{Ox11}: Let $n=8,\,k=4$ and
         \[
           \cA=\big\{\{1,2,3,4\},\,\{1,4,5,6\},\,\{2,3,5,6\},\,\{1,4,7,8\},\,\{2,3,7,8\}\big\}.
         \]
         The matroid $V_8:=M_{[8],\cA}$ is known as the Vamos matroid. It is not representable over any field
         \cite[p.~169, Ex.~7(e)]{Ox11}.
         In fact, it is the smallest such matroid.
         Choose any finite field~$\F$ and let $\xi$ be a bijection between $[n]$ and a fixed basis $\cB$ of $\F^n$.
         Then \cref{T-qPMClassMatr} and \cref{P-LinkMatrqMatr} tell us that the $q$-matroid $\cM_{8,\F_q,\cS}$,
         where $\cS=\{\subspace{\xi(A)}\mid A\in \cA\}$, is not $\F_{q^m}$-representable for any~$m\in\N$.
          A similar example can be constructed for $n=9,\,k=3$ and the non-Pappus matroid~\cite[Ex.~1.5.15]{Ox11}.
\item The Fano Matroid \cite[Ex.~1.5.7]{Ox11}: Let $n=7,\,k=3$ and
         \[
            \cA=\big\{\{1,2,4\},\,\{2,3,5\},\,\{3,4,7\},\,\{1,5,7\},\,\{2,6,7\},\,\{4,5,6\},\,\{1,3,6\}\big\},
         \]
         which is the set of lines in the Fano plane.
         The matroid $F_7:=M_{[7],\cA}$ is called the Fano matroid. It is not representable over any field of odd characteristic
         \cite[Prop.~6.4.8]{Ox11}.
         Thus, if~$q$ is odd then the corresponding $q$-matroid $\cM_{7,\F_q,\cS}$,
         where $\cS=\{\subspace{\xi(A)}\mid A\in \cA\}$, is not $\F_{q^m}$-representable for any~$m$ (here $\xi$ is a bijection between $[7]$ and a basis of $\F_q^7$).
\item The non-Fano Matroid: If one omits one of the sets in the collection~$\cA$ in~(b), one obtains the non-Fano matroid
        $F_7^-:=M_{[7],\cA'}$, which is not representable over any field of even characteristic~\cite[Prop.~6.4.8]{Ox11}.
        Thus we have the analogous conclusion as in~(b) for $q$-matroids with even~$q$.
\end{alphalist}
\end{exa}

It remains an open question whether any of the above $q$-matroids is $\F^{n\times m}$-representable.
However, we have the following smaller example of a $q$-matroid~$\cM$ that is not $\F^{n\times m}$-representable for any~$m$.
It has been shown very recently in \cite[Sec.~3.3]{CeJu21} that~$\cM$ is not $\F_{q^m}$-representable.
This shows that the converse of \cref{T-qPMClassMatr} is not true because every matroid over a groundset of cardinality~4
is representable over any field of size at least~$3$.
In fact, the $q$-matroid of the following theorem together with all possible choices of bases~$\cB$ in \cref{T-qPMClassMatr}
lead -- up to isomorphism -- to the uniform matroid of rank~$2$ or to the paving matroids  $M_{[4],\cA}$ with $\cA=\{\{1,2\}\}$ or $\cA=\{\{1,2\},\{3,4\}\}$ (see  the notation of \cref{P-MatroidExcSet}).
The latter two matroids are representable over every field.

\begin{theo}\label{E-NotRepr}
Let $\F=\F_2$ and consider
$\cS=\{V_0,V_1,V_2,V_3\}\subset\cV(\F^4)$, where
\[
    V_0=\subspace{1000,\,0100}, V_1=\subspace{0010,\,0001},  V_2=\subspace{1001,\,0111}, V_3=\subspace{1011,\,0110}.
\]
(We may also choose any other partial spread of size $4$).
Let $\cM=\cM_{4,\F,\cS}$ (see \cref{P-qMatroidExcSpace}), that is  $\cM=(\F^4,\rho)$, where
\[
  \rho(V)=1\text{ for }V\in\cS\ \text{ and }\ \rho(V)=\min\{2,\dim V\}\text{ otherwise}.
\]
Then $\cM$ is not $\F^{4\times m}$-representable for any $m\in\N$.
\end{theo}

\begin{proof}
Assume to the contrary that there exists $m\in\N$ and a rank-metric code $\cC\leq\F^{4\times m}$ such that
\[
    \rho(V)=\frac{\dim\cC-\dim\cC(V^\perp,\cc)}{m}\ \text{ for all } V\in\cV(\F^4).
\]
Using $V=\F^4$, we see that $\dim\cC=2m$.
Furthermore, the above values of $\rho(V)$ and the identity $\cS=\{V_0^\perp,\ldots,V_3^\perp\}$ lead to
\begin{equation}\label{e-CVc}
   \dim\cC(V,\cc)=\left\{\begin{array}{cl}0,&\text{if }\dim V\leq1 \text{ or $[\dim V=2$ and $V\not\in\cS]$},\\ m,&\text{if } V\in\cS\text{ or }\dim V=3.
   \end{array}\right.
\end{equation}
The conditions $\dim\cC(V_0,\cc)=\dim\cC(V_1,\cc)=m$ and $\dim\cC(V,\cc)=0$ if $\dim V=1$ imply that $\cC$ has a basis of the form
\begin{equation}\label{e-Basis}
   \begin{pmatrix}A_1\\0\end{pmatrix},\ldots,\begin{pmatrix}A_m\\0\end{pmatrix},
   \begin{pmatrix}0\\B_1\end{pmatrix},\ldots,\begin{pmatrix}0\\B_m\end{pmatrix},
\end{equation}
where $\cA=\subspace{A_1,\ldots,A_m}$ and $\cB=\subspace{B_1,\ldots,B_m}$ are MRD codes in $\F^{2\times m}$.
As for the spaces~$V_2$ and~$V_3$ note that
\[
   V_2=\rowsp\begin{pmatrix}1&0&0&1\\0&1&1&1\end{pmatrix}=\rowsp\big(I_2\mid S\T\big), \quad
   V_3=\rowsp\begin{pmatrix}1&0&1&1\\0&1&1&0\end{pmatrix}=\rowsp\big(I_2\mid T\T\big),
\]
where
\[
  S=\begin{pmatrix}0&1\\1&1\end{pmatrix},\  T=S^2=S^{-1}=\begin{pmatrix}1&1\\1&0\end{pmatrix}.
\]
Clearly, a matrix $M\in\cC$ is in $\cC(V_2,\cc)$ if and only if $M=\Smalltwomat{A}{SA}$ for some $A\in\cA$.
Hence there exist linearly independent matrices $\Smalltwomat{\tilde{A}_i}{S\tilde{A}_i},i=1,\ldots,m$, in~$\cC$.
But then $\tilde{A}_1,\ldots,\tilde{A}_m\in\cA$ must be linearly independent.
Thus $\cB=S\cA$. Using the space~$V_3$ we obtain similarly $\cB=T\cA$.
Hence $\cA=T^{-1}S\cA$, and the latter is $T\cA$. In other words,~$\cA$ is $T$-invariant.
Note that $\{0,I,T,T^2\}$ is the subfield~$\F_4$, and in particular, $T^2=I+T$.
All of this shows that $\cA$ is an $\F_4$-vector space (thus has even dimension over~$\F_2$) and
has an $\F_2$-basis of the form $A_1,\ldots,A_\ell,TA_1,\ldots, TA_\ell$, where $\ell=m/2$.
Using this basis of~$\cA$,~\eqref{e-Basis} reads as
\[
      \begin{pmatrix}A_1\\0\end{pmatrix},\ldots,\begin{pmatrix}A_\ell\\0\end{pmatrix},
      \begin{pmatrix}TA_1\\0\end{pmatrix},\ldots,\begin{pmatrix}TA_\ell\\0\end{pmatrix},
      \begin{pmatrix}0\\TA_1\end{pmatrix},\ldots,\begin{pmatrix}0\\TA_\ell\end{pmatrix},
      \begin{pmatrix}0\\(I+T)A_1\end{pmatrix},\ldots,\begin{pmatrix}0\\(I+T)A_\ell\end{pmatrix}.
\]
This shows that~$\cC$ contains the matrices $\Smalltwomat{A_i+TA_i}{A_i+TA_i}$ for $i=1,\ldots,\ell$, and therefore
$\dim\cC(V,\cc)\geq \ell$ for $V=\subspace{1010,0101}$. This contradicts~\eqref{e-CVc} and we conclude that
there is no code $\cC\leq\F^{4\times m}$ that represents the $q$-matroid $\cM$.
\end{proof}

It  is not yet clear whether there exists a $q$-matroid $\cM=(\F^n,\rho)$ and $m\in\N$ such that $\cM$ is not $\F_{q^m}$-representable  but
$\F^{n\times m}$-representable.

\section{Deletions and Contractions}\label{S-DelContr}
We define deletion and contraction for \qPM{}s and show that they correspond to puncturing and shortening of rank-metric codes.

\begin{defi}\label{D-Restrict}
Let $\cM=(E,\rho)$ be a \qPM{} and $X\in\cV(E)$.
\begin{alphalist}
\item Define $\rho|_X:\cV(X)\longrightarrow\Q_{\geq0},\ W\longmapsto \rho(W)$.
        Then $\cM|_X:=(X,\rho|_X)$ is a \qPM{} on~$X$ and is called the \emph{restriction} of~$\cM$ to~$X$.
\item Fix a non-degenerate symmetric bilinear form $\inner{\cdot}{\cdot}$ on~$E$. The restriction of~$\cM$ to $X^\perp$ is called the \emph{deletion} of~$X$ from~$\cM$ w.r.t.~$\inner{\cdot}{\cdot}$, and the resulting \qPM{}
        $(\cM|_{X^\perp},\rho|_{X^\perp})$ is denoted by $(\cM\setminus X)_{\langle\cdot\mid\cdot\rangle}$ or simply $\cM\setminus X$  if the
        bilinear form is clear from the context.
\end{alphalist}
\end{defi}

In~(a), it is clear that $\rho|_X$ satisfies (R1)--(R3) from \cref{D-PMatroid} and thus is a $q$-rank function.
In~(b), the deletion $(\cM\setminus X)_{\langle\cdot\mid\cdot\rangle}$ truly depends on $\inner{\cdot}{\cdot}$:
different choices of the bilinear form lead in general to non-equivalent \qPM{}s.
However, for any bilinear forms $\inner{\cdot}{\cdot}$ and $\langle\!\inner{\cdot}{\cdot}\!\rangle$ on~$E$
with orthogonal spaces $V^\perp$ and $V^{\pperp}$, respectively, we obviously have
$X^\perp=Y^{\pperp}$ for $Y=(X^\perp)^{\pperp}$ and thus $\cM|_{X^\perp}=\cM|_{Y^{\pperp}}$.

We now turn to contractions.

\begin{theo}\label{T-Restrict}
Let $\cM=(E,\rho)$ be a \qPM{} and $X\in\cV(E)$. Let $\pi:E\longrightarrow E/X$ be the canonical projection.
We define the map
\[
    \rho_{E/X}:\cV(E/X)\longrightarrow\Q_{\geq0},\quad V\longmapsto \rho(\pi^{-1}(V))-\rho(X).
\]
 Then $\cM/X:=(E/X,\rho_{E/X})$ is a \qPM{}, called the \emph{contraction} of~$X$ from~$\cM$.
\end{theo}

\begin{proof}
We have to show that $\rho_{E/X}$ satisfies (R1)--(R3).
\\
(R1) Let $V\in\cV(E/X)$. Then $X\leq \pi^{-1}(V)$ and thus the monotonicity of~$\rho$ implies $\rho_{E/X}(V)\geq0$.
Next, $\pi^{-1}(V)=X\oplus A$ for some $A\in\cV(E)$. Then $\dim A=\dim V$ and
submodularity of~$\rho$ yields $\rho_{E/X}(V)=\rho(\pi^{-1}(V))-\rho(X)\leq\rho(A)\leq\dim A=\dim V$.
\\
(R2) $V_1\leq V_2\leq E/X$ implies $\pi^{-1}(V_1)\leq\pi^{-1}(V_2)$ and therefore $\rho_{E/X}(V_1)\leq \rho_{E/X}(V_2)$.
\\
(R3) Let $V,\,W\leq E/X$.
Then $\pi^{-1}(V+W)=\pi^{-1}(V)+\pi^{-1}(W)$ and $\pi^{-1}(V\cap W)=\pi^{-1}(V)\cap\pi^{-1}(W)$ and
therefore
\begin{align*}
  \rho_{E/X}(V\!+\!W)\!+\!\rho_{E/X}(V\cap W)&=\rho\big(\pi^{-1}(V)+\pi^{-1}(W)\big)-\rho(X)+\rho\big(\pi^{-1}(V)\!\cap\!\pi^{-1}(W)\big)-\rho(X)\\
      &\leq \rho(\pi^{-1}(V))+\rho(\pi^{-1}(W))-2\rho(X)= \rho_{E/X}(V)+ \rho_{E/X}(W).
      \qedhere
\end{align*}
\end{proof}

Deletion and contraction are mutually dual in the following sense.
Since duality is involved, we need to pay special attention to the choice of the non-degenerate symmetric bilinear form.
For $\dim X=1$, the following result also appears in~\cite[Thm.~60]{JuPe18}.
However, the proof given there does not apply if $X\leq X^\perp$ and yields a weaker form of equivalence between the matroids.
Also in \cite[Lem.~12]{BCIJ21} the result below is proven for  a weaker form of equivalence.
Recall from \cref{D-EquivMatroid} that in this paper equivalence of \qPM{}s is based on linear isomorphisms between the ground spaces.

\begin{theo}\label{T-DeletContra}
Let $\cM=(E,\rho)$ be a \qPM{} and $X\in\cV(E)$. Then
\[
      (\cM\setminus X)^*\approx\cM^*/X\ \text{ and }\ \cM\setminus X\approx(\cM^*/X)^*.
\]
\end{theo}

\begin{proof}
We will show the first equivalence. The second one follows from from biduality; see \cref{T-DualqPM} and
\cref{P-EquivDual}.
We need some preparation.
Recall from \cref{T-DualqPM} that the dual of a \qPM{} depends on the choice of the non-degenerate symmetric bilinear
form (NSBF) and that different choices lead to equivalent dual \qPM{}s.
Hence for the first equivalence we need NSBFs on~$E$ and on~$X^\perp$, the latter being the ground space of
$\cM\setminus X$.
Note that if $\inner{\cdot}{\cdot}$ is an NSBF on~$E$,
then the restriction of $\inner{\cdot}{\cdot}$ to the resulting orthogonal~$X^\perp$ is in general degenerate.
For this reason we proceed as follows.
Choose a subspace $Y\in\cV(E)$ such that $X\oplus Y=E$.
Choose NSBFs $\inner{\cdot}{\cdot}_X$ on~$X$ and $\inner{\cdot}{\cdot}_Y$ on~$Y$ and define
\[
   \inner{x_1+y_1}{x_2+y_2}:=\inner{x_1}{x_2}_X+\inner{y_1}{y_2}_Y \ \text{ for all }\ x_1,x_2\in X,\,y_1,y_2\in Y.
\]
It is easy to verify that $\inner{\cdot}{\cdot}$ is a NSBF on~$E$.
Denoting the resulting orthogonal of a subspace $V\in\cV(E)$ by~$V^\perp$, we observe $X^\perp=Y$ and thus
$X\oplus X^\perp=E$.
Furthermore, for any subspace $Z\leq X^\perp$ we have
\begin{equation}\label{e-ZperpY}
    Z^{\pperp}=Z^\perp\cap X^\perp,
\end{equation}
where $Z^{\pperp}$ denotes the orthogonal of~$Z$ in~$X^\perp$ w.r.t.\ $\inner{\cdot}{\cdot}_Y$.
Now we have compatible NSBFs on~$E$ and~$X^\perp$ and can turn to the stated equivalence
$(\cM\setminus X)^*\approx\cM^*/X$.

Note that we have a well-defined isomorphism
\[
   \xi: E/X\longrightarrow X^\perp,\quad v+X\longmapsto \hat{x},
\]
where $v=x+\hat{x}$ is the unique decomposition of~$v$ into $x\in X$ and $\hat{x}\in X^{\perp}$.
We show that
\begin{equation}\label{e-rho*rho}
      (\rho^*)_{E/X}(V)= (\rho|_{X^\perp})^*(\xi(V))\ \text{ for all }\ V\in\cV(E/X).
\end{equation}
Let $\pi:E\longrightarrow E/X$ be the canonical projection and $V\in\cV(E/X)$.
Then  $\pi^{-1}(V)=\xi(V)\oplus X$ and thus $\pi^{-1}(V)^\perp=\xi(V)^\perp\cap X^\perp$.
Now we compute
\begin{align*}
      (\rho^*)_{E/X}(V)&=\rho^*(\pi^{-1}(V))-\rho^*(X)\nonumber \\
         &=\dim \pi^{-1}(V)+\rho(\pi^{-1}(V)^\perp)-\rho(E)-\big(\dim X+\rho(X^\perp)-\rho(E)\big)\nonumber \\
         &=\dim \pi^{-1}(V)-\dim X+\rho(\pi^{-1}(V)^\perp)-\rho(X^\perp)\nonumber \\
         &=\dim\xi(V)+\rho(\xi(V)^\perp\cap X^\perp)-\rho(X^\perp)\\
         &=\dim\xi(V)+\rho|_{X^\perp}(\xi(V)^\perp\cap X^\perp)-\rho|_{X^\perp}(X^\perp)\\
         &=\dim\xi(V)+\rho|_{X^\perp}(\xi(V)^{\pperp})-\rho|_{X^\perp}(X^\perp)\\
         &=(\rho|_{X^\perp})^*(\xi(V)),
\end{align*}
where the penultimate step follows from~\eqref{e-ZperpY} and the last step is the very definition of $(\rho|_{X^\perp})^*(\xi(V))$
for the chosen NSBF on~$X^\perp$.
\end{proof}

The rest of this section is devoted to the relation between deletion and contraction of \qPM{}s induced by rank-metric codes and puncturing and shortening of the codes.
We focus on row puncturing and shortening, which will correspond to deletion and contraction of the associated column \qPM{}.
The following terminology is from \cite[Sec.~3]{ByRa17}.
For the rest of this section,~$X^\perp$ denotes the orthogonal of~$X\leq\F^n$ w.r.t.\ the standard dot product.

\begin{defi}\label{D-Shorten}
Let $u\in[n]$ and $\pi_u:\Fnm\longrightarrow \F^{(n-u)\times m}$
be the projection onto the last $n-u$ rows.
Let  $\cC\leq\F^{n\times m}$ be a rank-metric code.
We define
\[
   \cC_u=\{M\in\cC\mid M_{ij}=0\text{ for }i\leq u\},
\]
that is,~$\cC_u$ is the subcode consisting of all matrices of~$\cC$ whose first~$u$ rows are zero.
Let $A\in\GL_n(\F)$.
\begin{alphalist}
\item The \emph{puncturing} of~$\cC$ w.r.t.~$A$ and~$u$ is defined as the code
         \[
                \Pi(\cC,A,u)=\pi_u(A\cC)\leq\F^{(n-u)\times m}.
         \]
\item The \emph{shortening} of~$\cC$ w.r.t.~$A$ and~$u$ is defined as the code
        \[
           \Sigma(\cC,A,u)=\pi_u((A\cC)_u)\leq\F^{(n-u)\times m}.
        \]
\end{alphalist}
\end{defi}

In \cref{D-RMCBasics}(d) we called the space $\cC(V,\cc)$ a shortening of~$\cC$.
This is consistent because $\cC(V,\cc)$ is isomorphic (even isometric w.r.t.\ the rank metric) to $\Sigma(\cC,A,u)$, where
$u=\dim V$ and $A=(A_1\mid A_2)\T\in\GL_n(\F)$ such that $\colsp(A_1)=V^\perp$.

Now we are ready to relate deletion and contraction of \qPM{}s induced by rank-metric codes to puncturing and shortening of the codes, respectively.
The following result is a direct $q$-analogue of the well-known relation between matroids and linear block codes; see
\cite[Sec.~1.6.4]{JuPe13}.

\begin{theo}\label{T-DelContrCode}
Let  $\cC\leq\F^{n\times m}$ be a rank-metric code and $\cM=\cM_\cc(\cC)$ be the associated column \qPM{}.
Let $X\in\cV(\F^n)$ and set $\dim X=u$.
Choose any matrix $A\in\GL_n(\F)$ of the form $A=\left(\begin{smallmatrix}B\\[.3ex]D\end{smallmatrix}\right)$, where $D\in\F^{(n-u)\times n}$ is such that
$\rowsp(D)=X^\perp$.
Then
  \[
             \cM\setminus X\approx\cM_\cc(\Pi(\cC,A,u))\ \text{ and }\
             \cM/X\approx \cM_\cc(\Sigma(\cC,(A^{-1})\T,u)).
  \]
\end{theo}

\begin{proof}
We start with the first equivalence.
Let $\rho$ be the rank function of~$\cM_\cc(\cC)$ and $\rho'$ be the rank function of $\cM_\cc(\cC')$, where $\cC'=\Pi(\cC,A,u)$.
Thus
\[
    \rho':\cV(\F^{n-u})\longrightarrow\Q_{\geq0},\ V\longmapsto \frac{\dim\cC'-\dim\cC'(V^\perp,\rr)}{m}.
\]
Consider the isomorphism
\[
    \psi:\F^{n-u}\longrightarrow X^\perp,\ v\longmapsto vD.
\]
Then $\psi(V)\leq X^\perp$ for all $V\in\cV(\F^{n-u})$ and therefore $\rho_{\cM\setminus X}(\psi(V))=\rho(\psi(V))$ by the very definition of deletion.
Now the desired equivalence $\cM\setminus X\approx\cM_\cc(\Pi(\cC,A,u))$ follows if we can show that
\begin{equation}\label{e-rhorho'}
    \rho'(V)=\rho(\psi(V))\text{ for all }V\in\cV(\F^{n-u}).
\end{equation}
In order to do so, let $V\in\cV(\F^{n-u})$ and let $Z\in\F^{t\times(n-u)}$ be such that $V=\rowsp(Z)$.
From the very definition of the matrices involved we obtain
\begin{align*}
     \cC(X,\cc)&=\{M\in\cC\mid \colsp(M)\leq X\}=\{M\in\cC\mid DM=0\},\\[.5ex]
     \cC(\psi(V)^\perp,\cc)&=\{M\in\cC\mid \colsp(M)\leq\psi(V)^\perp\}=\{M\in\cC\mid ZDM=0\},\\[.5ex]
     \cC'(V^\perp,\cc)&=\{N\in\cC'\mid \colsp(N)\leq V^\perp\}=\{N\in\cC'\mid ZN=0\}.
\end{align*}
Furthermore, we have the surjective linear map
\[
   \xi_D:\cC\longrightarrow \cC',\ M\longmapsto DM
\]
and observe that $\cC(X,\cc)=\ker\xi_D\leq \cC(\psi(V)^\perp,\cc)$ as well as
$\xi_D\big(\cC(\psi(V)^\perp,\cc)\big)=\cC'(V^\perp,\cc)$. Hence we conclude
\[
   \dim\cC-\dim\cC'=\dim\cC(X,\cc)=
   \dim\cC(\psi(V)^\perp,\cc)-\dim\cC'(V^\perp,\cc).
\]
As a consequence,
\[
    \frac{\dim\cC'-\dim\cC'(V^\perp,\cc)}{m}=\frac{\dim\cC-\dim\cC(\psi(V)^\perp,\cc)}{m},
\]
and this establishes~\eqref{e-rhorho'}.
\\
We now turn to the second equivalence.
Applying the first equivalence to the code~$\cC^\perp$ we obtain
$\cM_\cc(\cC^\perp)\setminus X\approx\cM_\cc(\Pi(\cC^\perp,A,u))$.
With the aid of Theorems~\ref{T-DeletContra} and \ref{T-TraceDual} this leads to
\[
   \cM_\cc(\cC)/X\approx (\cM_\cc(\cC^\perp)\setminus X)^*\approx\cM_\cc(\Pi(\cC^\perp,A,u))^*=\cM_\cc(\Pi(\cC^\perp,A,u)^\perp).
\]
Now the desired equivalence follows from $\Pi(\cC^\perp,A,u)^\perp=\Sigma(\cC,(A^{-1})\T,u)$, which has been proven in
\cite[Thm.~3.5]{ByRa17} (and is true for $n\leq m$ and $n>m$).
\end{proof}

\section{Flats of a $q$-Polymatroid}\label{S-Flats}

In this section we return to general \qPM{}s and introduce flats.
As in (classical) matroid theory, a flat is, by definition, an inclusion-maximal space for a given rank.
Flats naturally come with a closure operator.
We derive some basis properties of the closure operator and flats and will
illustrate that -- just like for classical polymatroids -- the lattice of flats is not semimodular and (without the associated rank values) does not fully determine the \qPM.
In fact, a \qPM{} may have the same flats as a $q$-matroid without being a $q$-matroid itself.
On the positive side, we show that the generalized weights of a rank-metric code are fully determined by the flats of the
associated \qPM{}, and the rank-distance is determined by the hyperplanes.

Throughout let $\cM=(E,\rho)$ be a \qPM.

\begin{defi}\label{D-Flat}
A space $F\in\cV(E)$ is called a \emph{flat} of~$\cM$ if
\[
    \rho(F+\subspace{x})>\rho(F)\text{ for all  }x\in E\setminus F.
\]
We denote the collection of flats by $\cF(\cM)$, or simply~$\cF$.
A flat~$H$ is called a \emph{hyperplane} if there is no flat strictly between~$H$ and~$E$.
Furthermore, we define the \emph{closure operator} of~$\cM$ as
\[
   \cl:\cV(E)\longrightarrow\cV(E),\quad V\longmapsto
   \sum_{   \genfrac{}{}{0pt}{1}{\dim X=1}{\rho(V+X)=\rho(V)}     }\hspace*{-.7em}X.
\]

\end{defi}

Clearly, $E\in\cF(\cM)$, and $0\in\cF(\cM)$ if and only if  $\cM$ has no loops (by definition, a loop is a $1$-space of rank zero).
With the aid of \cref{P-RankVx}(a) we obtain immediately
\begin{equation}\label{e-ClosureRho}
   \rho(V)=\rho(\cl(V)) \text{ for all }V\in\cV(E).
\end{equation}

Flats can be regarded as `rank-closed' subspaces.

\begin{prop}\label{P-FlatClosureA}
A subspace $F\in\cV(E)$ is a flat of~$\cM$ if and only if $F=\cl(F)$.
\end{prop}

\begin{proof}
If~$F=\cl(F)$ and $x\not\in F$, then $\rho(F+\subspace{x})>\rho(F)$. Hence $F$ is a flat.
Conversely, let~$F$ be a flat and let $x\in\cl(F)$.
Then $\rho(F+\subspace{x})=\rho(F)$, and thus~$x\in F$. Hence $F=\cl(F)$.
\end{proof}

The closure operator satisfies the following properties.

\begin{theo}\label{T-Closure}
Let $V,W\in\cV(E)$. Then
\begin{mylist3}
\item[(CL1)\hfill] $V\leq\cl(V)$.
\item[(CL2)\hfill] If $V\leq W$, then $\cl(V)\leq \cl(W)$.
\item[(CL3)\hfill] $\cl(V)=\cl(\cl(V))$.
\end{mylist3}
\end{theo}

\begin{proof}
(CL1) is obvious.
\\
(CL2) Let $V\leq W$ and let $x\in\cl(V)$. We want to show that $x\in\cl(W)$.
This is clear if $x\in V$, and thus we assume that $x\not\in V$.
Choose $U\in\cV(E)$ such that $W=V\oplus U$.
Then
\begin{align*}
   \rho(W+\subspace{x})&=\rho((V\oplus U)+(V\oplus\subspace{x}))\\
      &\leq\rho(V\oplus U)+\rho(V\oplus\subspace{x})-\rho((V\oplus U)\cap(V\oplus\subspace{x}))\\
      &=\rho(W)+\rho(V)-\rho((V\oplus U)\cap(V\oplus\subspace{x}))\\
      &= \rho(W)+\rho(V)-\rho(V)=\rho(W),
\end{align*}
where the penultimate step follows from $V\leq (V\oplus U)\cap(V\oplus\subspace{x})\leq V\oplus\subspace{x}$ together with $\rho(V\oplus\subspace{x})=\rho(V)$.
All of this shows that $\rho(W+\subspace{x})=\rho(W)$ and therefore $x\in\cl(W)$.
\\
(CL3) Let $V\in\cV(E)$. Then $\cl(V)\leq\cl(\cl(V))$ follows from (CL1) since $\cl(V)$ is a subspace.
For the converse let $x\in\cl(\cl(V))$.
With the aid of (CL1) and \eqref{e-ClosureRho} we obtain
\[
  \rho(V)\leq\rho(V+\subspace{x})\leq \rho(\cl(V)+\subspace{x})=\rho(\cl(V))=\rho(V).
\]
Thus we have equality across, and this shows that $x\in\cl(V)$.
\end{proof}

Flats satisfy the following simple properties.

\begin{theo}\label{T-FlatProperties}
Let $\cF$ be the collection of flats of~$\cM$. Then
\begin{mylist2}
\item[(F1)\hfill]  $E\in\cF$.
\item[(F2)\hfill] If $F_1,F_2\in\cF$, then $F_1\cap F_2\in\cF$.
\end{mylist2}
\end{theo}

\begin{proof}
(F1) is clear from \cref{D-Flat}.
For (F2) let $F_1,F_2\in\cF$.
Then (CL2) yields
$\cl(F_1\cap F_2)\subseteq\cl(F_1)\cap\cl(F_2)=F_1\cap F_2$.
Thus $F_1\cap F_2\in\cF$ thanks to \cref{P-FlatClosureA}.
\end{proof}

From (CL2), \cref{P-FlatClosureA}, and (F2) we immediately obtain

\begin{prop}\label{P-FlatClosure}
Any $V\in\cV(E)$ satisfies $\cl(V)=\bigcap_{\genfrac{}{}{0pt}{2}{F\in\cF}{V\leq F}}F$.
\end{prop}

For $q$-matroids,  the closure operator and the flats satisfy further properties, which we briefly discuss.
Note first that for any \qPM{}~$\cM$ the collection~$\cF$ of flats forms a lattice under inclusion:
the \emph{meet} and \emph{join} of two flats $F_1,\,F_2\in\cF$ are $F_1\wedge F_2:=F_1\cap F_2$ and
$F_1\vee F_2:=\cl(F_1+F_2)$, respectively.
The flat $\cl(0)=\bigcap_{F\in\cF}F$ is the unique minimal element of the lattice~$\cF$.
Furthermore, we say that $F_2$ \emph{covers} $F_1$ if $F_1\lneq F_2$ and there exists no $F\in\cF$ such that
$F_1\lneq F\lneq F_2$.
For instance, the hyperplanes of $\cF $ are the flats covered by~$E$.

The first statement below has been proven in \cite[Thm.~68]{JuPe18} whereas the second one can be found in
\cite[Prop.~26, Thm.~28, Cor.~29, Thm.~31]{BCIJS20} and \cite[Cor.~70]{BCJ22}.
Note that these result generalize the according properties for matroids; see \cite[pp.~26 and 32]{Ox11}.

\begin{theo}
\label{T-qMatroids}
Let $\cM=(E,\rho)$ be a $q$-matroid with closure operator~$\cl$ and collection of flats~$\cF$.
Then:
\begin{mylist3}
\item[(CL4)\hfill] MacLane-Steinitz exchange axiom:
For all $V\in\cV(E)$ and all vectors $x,y\in E\setminus\cl(V)$ we have
       $\subspace{y}\leq\cl(V+\subspace{x})\Longleftrightarrow \subspace{x}\leq\cl(V+\subspace{y})$.
\item[(F3)\hfill] For all $F\in\cF$ and vectors $x\in E\setminus F$ there exists a unique cover $F'\in\cF$ of~$F$ such that $\subspace{x}\leq F'$.
\end{mylist3}
Furthermore,~$\cF$ is a semimodular lattice (that is, for all $F_1,F_2\in\cF$ such that
$F_1$ covers $F_1\wedge F_2$, the flat $F_1\vee F_2$ covers~$F_2$).
As a consequence, all maximal chains between any two fixed elements in~$\cF$ have the same length.
Finally, if~$\cF$ is a collection of subspaces of~$E$ satisfying (F1)--(F3), then~$\cF$ is the collection of flats of the $q$-matroid
$(E,\rho)$ defined by $\rho(V)=\text{h}(\cl(V))$, where the latter is the length of a maximal chain from~$\cl(0)$
to $\cl(V)$.
It follows that all hyperplanes have the same rank.
\end{theo}

The first of the following examples illustrates that neither (F3) nor (CL4) is true for \qPM{}s.
This comes to no surprise, since (F3) and (CL4) do not hold for classical polymatroids either.
The second example shows that the lattice of flats (without  rank values) carries insufficient information to
distinguish \qPM{}s from $q$-matroids.
Of course, thanks to \cref{P-FlatClosure} and~\eqref{e-ClosureRho} the flats  together with their
rank values uniquely determine the \qPM{}.

\begin{exa}\label{E-FlatsNotF3}
\begin{alphalist}
\item Consider the code $\cC_1=\subspace{A_1,\dots,A_5}\leq\F_2^{5\times2}$ from \cref{E-RowMatMRD}.
       It can be verified that $\cM_\cc(\cC)$ has 81 flats and 29 hyperplanes.
       Furthermore, the collection of flats does not satisfy (CL4) and (F3).
       Unsurprisingly, the hyperplane axioms known for $q$-matroids (see \cite[Def.~12]{BCJ22}) do not hold either, and in particular,
       the hyperplanes do not all have the same rank.

\item Let $\F=\F_2$ and $\cC=\subspace{E_{12}+E_{23}+E_{33},\,E_{12}+E_{13}+E_{33},\, E_{12}+E_{13}+E_{21}+E_{32}}\leq \F^{3\times 3}$.
         It is straightforward to verify that the associated column \qPM{} $\cM=(\F^3,\,\rho_\cc)$ is not a $q$-matroid but has the same collection of flats as the $q$-matroid
         $\cM_G$, where
          \[
                G=\begin{pmatrix}1&0&\alpha\\0&1&\alpha\end{pmatrix}\in\F_4^{2\times3}.
          \]
          As a consequence, $\cF(\cM)$ and the closure operator satisfy (F1)--(F3) and (CL1)--(CL4).
          \end{alphalist}
\end{exa}

The rest of this section is devoted to showing that the generalized weights of a rank-metric code in $\Fnm$ can be determined by
the flats of the associated column \qPM{} if $m> n$. For $m<n$, the corresponding result is true for the row
\qPM{}, and for $m=n$ both \qPM{}s are needed.
In order to be aligned with most of the literature on rank-metric codes, we will assume $m\geq n$.
For further details on  generalized weights and anticodes we refer to \cite{Ra16b}, from which the next definition is taken.

\newpage
\begin{defi}[\mbox{\cite[Def.~22 and~23]{Ra16b}}]\label{D-GenWeights}
Let $m\geq n$ and $\cC\leq\Fnm$ be a  code.
\begin{alphalist}
\item We define $\maxrk(\cC)=\max\{\rk(M)\mid M\in\cC\}$.
         We call the code~$\cC$ an \emph{optimal anticode} if $\dim\cC=m\,\maxrk(\cC)$.
\item For $i=1,\ldots,\dim\cC$ the \emph{$i$-th generalized weight} of~$\cC$ is defined as
        \[
            a_i(\cC)=\frac{1}{m}\min\{\dim(\cA)\mid \cA\leq\Fnm\text{ is an optimal anticode}, \dim(\cC\cap\cA)\geq i\}.
        \]
\end{alphalist}
\end{defi}

In \cite{GJLR19} it has been shown that the generalized weights can be computed from the associated \qPM{}s of the code.
We will derive the generalized weights in a slightly different form, and, since we need the ideas of the proof later on, we also briefly
sketch the proof.
In \cite[Thm.~7.2]{GJLR19} the result has been used to characterize optimal anticodes with the aid of \qPM{}s.

\begin{theo}[\mbox{\cite[Thm.~7.1]{GJLR19}}]\label{T-GenWeightqPM}
Let $m\geq n$ and $\cC\leq\Fnm$ be a  code.
Define
\begin{align*}
    a_i^\cc(\cC)&=n-\max\{\dim V\mid V\in\cV(\F^n),\,\rho_\cc(V)\leq\rho_\cc(\F^n)-i/m\},\\[.6ex]
    a_i^\rr(\cC)&=m-\max\{\dim W\mid W\in\cV(\F^m),\,\rho_\rr(W)\leq\rho_\rr(\F^m)-i/n\}.
\end{align*}
Then for all $i=1,\ldots,\dim\cC$
\[
    a_i(\cC)=\begin{cases} \qquad a_i^\cc(\cC),&\text{if }m>n,\\[.6ex]\ \min\{a_i^{\cc}(\cC),a_i^{\rr}(\cC)\},&\text{if }m=n.\end{cases}
\]
\end{theo}

\begin{proof}
Let first $m> n$. It is well known \cite[Thm.~3]{Mes85} that the optimal anticodes in $\Fnm$ are the spaces of the form $\Fnm(V,\cc)$ for any
$V\in\cV(\F^n)$.
Furthermore $\dim\Fnm(V,\cc)=m\dim V$. Since $\cC\cap\Fnm(V^\perp,\cc)=\cC(V^\perp,\cc)$ and
$m^{-1}\dim\cC=\rho_\cc(\F^n)$, the very definition of~$\rho_\cc$ leads to
\begin{equation}\label{e-cCcA}
   \dim(\cC\cap\Fnm(V^\perp,\cc))\geq i\Longleftrightarrow \rho_\cc(V)\leq\rho_\cc(\F^n)-i/m.
\end{equation}
Now the statement follows from $n-\dim V=\dim V^\perp=m^{-1}\dim \Fnm(V^\perp,\cc)$.
In the case $m=n$, the optimal anticodes are of the form $\F^{m\times m}(V,\cc)$ and $\F^{m\times m}(V,\rr)$, where $V\in\cV(\F^m)$.
This implies that we have to take the row and column \qPM{} into account, and again the statement follows.
\end{proof}

Let us briefly relate~$a_i^\cc(\cC)$ to certain invariants introduced in the literature.

\begin{rem}\label{R-TauIneq}
\begin{alphalist}
\item In \cite[Def.~10]{GhJo20} the authors introduce generalized weights of \qPM{}s.
Since their definition of \qPM{}s differs from ours, we need to be careful when comparing the two notions of generalized weights.
Let us start with a \qPM{} $\cM=(\F^n,\rho)$ in our sense and let~$\mu$ be a denominator of~$\cM$.
Set  $\tau:=\mu\rho$.
Then $\cM':=(\F^n,\tau)$ is a $(q,\mu)$-polymatroid in the sense of \cite[Def.~1]{GhJo20}.
In  \cite[Def.~10]{GhJo20} the authors define the $i$-th generalized weight of~$\cM'$ as
$d_i(\cM')=\min\{\dim X\mid \tau(\F^n)-\tau(X^\perp)\geq i\}$.
Suppose now that $\cM=\cM_\cc(\cC)$ for some rank-metric code $\cC\leq\Fnm$.
Then $\mu:=m$ is a denominator and  the equivalence $\rho(V)\leq\rho(\F^n)-i/m\Longleftrightarrow\tau(\F^n)-\tau(V)\geq i$
shows that $d_i(\cM')=a^{i}_\cc(\cC)$ for all~$i$.
\item In \cite[Thm.~8]{BMS20} the authors associate a $q$-demimatroid to a rank-metric code.
        A $q$-demimatroid is a generalization of a \qPM{} that does not require submodularity, but captures a certain duality relation;
        see \cite[p.~1505]{BMS20}.
        More precisely, to the code $\cC\leq\F^{n\times m}$  they associate the $q$-demimatroid $(\F^n,s,t)$ defined by
         $s(V)=\dim\cC(V,\cc)$ and $t(V)=\dim\cC^{\perp}(V,\cc)$.
        Neither of these functions is submodular, but they satisfy the required duality relation
        thanks to \eqref{e-DimCV}.
        In \cite[p.~1507]{BMS20} the authors introduce various combinatorial invariants of such $q$-demimatroids,
        and it is straightforward to verify that $a_i^\cc(\cC)$ equals their $\sigma_i$ for all $i$.
       The power of their approach comes to light in particular in \cite[Sec.~4.2 and~4.3]{BMS20}, where the authors provide a very
       nice and short proof of the Wei duality for the generalized weights of a rank-metric code with the aid of the associated $q$-demimatroid.
       \end{alphalist}
 \end{rem}

We return to \cref{T-GenWeightqPM} and consider the case $m=n$ in further detail.
As we show next, for  $i=1$ we have $a_1(\cC)=a_1^\cc(\cC)=a_1^\rr(\cC)$.

\begin{cor}\label{C-drkAntiCode}
Let $\cC\leq\F^{m\times m}$ be a  code. Then $a_1^\cc(\cC)=a_1^\rr(\cC)$ and therefore
\begin{align*}
  \drk(\cC)&=m-\max\{\dim V\mid V\in\cV(\F^m),\,\rho_\cc(V)\leq\rho_\cc(\F^m)-1/m\}\\
      &=m-\max\{\dim V\mid V\in\cV(\F^m),\,\rho_\rr(V)\leq\rho_\rr(\F^m)-1/m\}.
\end{align*}
\end{cor}

Recall from \cref{P-RMCMatroid} that $\rho_\cc(\F^m)=\rho_\rr(\F^m)=m^{-1}\dim\cC$ for any code $\cC\leq\F^{m\times m}$.
Thus the above inequalities can be written as $\rho_\cc(V)\leq(\dim\cC-1)/m$ and $\rho_\rr(V)\leq(\dim\cC-1)/m$.

\begin{proof}
It is well known \cite[Thm.~30]{Ra16b} that $\drk(\cC)=a_1(\cC)$, the first generalized weight.
Hence it suffices to show $a_1^\cc(\cC)=a_1^\rr(\cC)$.
Let $d=\drk(\cC)$ and $M\in\cC$ be such that $\rk(M)=d$.
Let $V=\colsp(M)$ and $W=\rowsp(M)$ and set $\cA=\F^{m\times m}(V,\cc)$ and $\cA'=\F^{m\times m}(W,\rr)$.
Then $\cA$ and~$\cA'$ are optimal anticodes of dimension~$md$ and $\cC\cap\cA\neq0\neq\cC\cap\cA'$.
Thus~\eqref{e-cCcA} yields $\rho_\cc(V^\perp)\leq\rho_\cc(\F^m)-1/m$ and, similarly, $\rho_\rr(W^\perp)\leq\rho_\rr(\F^m)-1/m$.
Since~$V$ and~$W$ are clearly of minimal dimension satisfying $\cC\cap\cA\neq0\neq\cC\cap\cA'$, all of
this shows that $a_1^\cc(\cC)=a_1^\rr(\cC)$.
\end{proof}

For $i>1$ and $\cC\leq\F^{m\times m}$ we have in general $a_i^{\cc}(\cC)\neq a_i^{\rr}(\cC)$, and it depends on~$i$ which of the two is the minimum.
Our next example shows that this is even the case for \emph{vector rank-metric codes} (that is, $\F_{q^m}$-linear codes in $\F_{q^m}^m$).
As a consequence, the generalized weights in \cref{D-GenWeights} do not coincide with the generalized weights introduced in \cite{JuPe17}.
The example also addresses a small oversight in \cite{Ra16b} and shows that \cite[Thm.~28]{Ra16b} is only true for codes in $\F_{q^m}^n$ with $n<m$.

\begin{exa}\label{E-FqmLinearGenWeights}
Let $q=2$ and $m=6$. Let $\omega\in\F_{2^6}$ be a primitive element with minimal polynomial $f=x^6+x^4+x^3+x+1$.
Consider the generator matrix
\[
   G=\begin{pmatrix}1&0&0&\omega^9&0&0\\0&1&0&0&\omega^{18}&0\\0&0&1&0&0&\omega^{18}\end{pmatrix}
   \in\F_{2^6}^{3\times 6}.
\]
Note that, in fact, the entries of~$G$ are in the subfield~$\F_{2^3}$.
Using the isomorphism~$\Psi$ from~\eqref{e-CoordMap}, we obtain the rank-metric code
$\Psi(\rowsp_{\F_{2^6}}(G))=\cC:=\subspace{A_j\Delta_f^i\mid j=1,2,3,\,i=0,\ldots,m-1}$,
where~$\Delta_f$ is as in \eqref{e-Deltaf} and $A_j$ is the image of the $j$-th row of~$G$ under~$\Psi$.
Thus $\dim_{\F_2}\cC=18$.
By construction, the row spaces of the matrices $A_1,A_2,A_3$ are contained in the $3$-dimensional subspace
$\hat{V}:=\Psi(\F_{2^3})\leq\F_2^6$.
It turns out that $\dim\cC(\hat{V},\rr)=9$ (which implies $a_9^\rr\leq 3$).
SageMath computations lead to the following data
\[
   \begin{array}{|c||c|c|c|c|c|c|c|c|c|c|c|c|c|c|c|c|c|c|}
   \hline
      i &1&2&3&4&5&6&7&8&9&10&11&12&13&14&15&16&17&18\\\hline\hline
     a_i^\cc &2&2&2&2&2&2&4&4&4&4&4&4&6&6&6&6&6&6 \\\hline
     a_i^\rr &2&2&2&3&3&3&3&3&3&5&5&5&6&6&6&6&6&6  \\\hline
    \end{array}
\]
We see that, for instance $a_9^\rr<a_9^\cc$, whereas $a_{10}^\rr>a_{10}^\cc$.
This  implies that the generalized weights $a_i(\cC)$  in general
do not coincide with $a_i^\cc(\cC)$.
In  \cite{JuPe17} the latter are defined to be the generalized weights of~$\cC$
(see \cite[Cor.~4.4, Thm.~5.4]{JuPe17} and \cite[Thm.~18]{Ra16b}).
This shows that the definitions of generalized weights in \cite[Def.~23]{Ra16b} and \cite{JuPe17}
do not agree for $\F_{q^m}$-linear codes in $\F_{q^m}^m$.
\end{exa}

We return to general codes in $\F^{n\times m}$.
Our next result shows that the generalized weights of a rank-metric code are determined by the flats of the \qPM.
It is straightforward to check that this is the analogue of \cite[Thm.~3]{Bri07}, where the generalized weights of a
linear block code are determined via cocircuits of suitable truncations of the associated matroid.
The complements of these cocircuits are, by definition, flats of a certain rank in the original matroid.
The only difference to our result below is the inequality (rather than equality) in $\rho_\cc(V)\leq\rho_\cc(\F^n)-i/m$
below. This is needed because -- in contrast to matroids and $q$-matroids -- for \qPM{}s equality may not be attained.

\begin{theo}\label{T-GenWeightqPMFlats}
Let $m\geq n$ and $\cC\leq\Fnm$ be a  code.
Let $\cF_\cc$ and $\cF_\rr$ be the sets of flats of the \qPM{}s $\cM_\cc(\cC)=(\F^n,\rho_\cc)$ and
$\cM_\rr(\cC)=(\F^m,\rho_\rr)$, respectively.
Define
\begin{align*}
    b_i^\cc(\cC)&=n-\max\{\dim V\mid V\in\cF_\cc,\,\rho_\cc(V)\leq\rho_\cc(\F^n)-i/m\},\\[.6ex]
    b_i^\rr(\cC)&=m-\max\{\dim W\mid W\in\cF_\rr,\,\rho_\rr(W)\leq\rho_\rr(\F^m)-i/n\}.
\end{align*}
Then for all $i=1,\ldots,\dim\cC$
\[
    a_i(\cC)=\begin{cases} \qquad b_i^\cc(\cC),&\text{if }m>n,\\[.6ex] \min\{b_i^{\cc}(\cC),b_i^{\rr}(\cC)\},&\text{if }m=n.\end{cases}
\]
Furthermore, if $n=m$ then $b_1^{\cc}(\cC)=b_1^{\rr}(\cC)$.
\end{theo}

\begin{proof}
We will show $b_i^\cc(\cC)=a_i^\cc(\cC)$ for all~$i$. The proof for $b_i^\rr(\cC)=a_i^\rr(\cC)$ is analogous.
Clearly, $a_i^\cc(\cC)\leq b_i^\cc(\cC)$ for all~$i$.
For the converse inequality, let $V\in\cV(\F^n)$ such that $\rho_\cc(V)\leq\rho_\cc(\F^n)-i/m$.
Consider $\cl(V)$, the closure of~$V$.
Then $\dim V\leq \dim\cl(V)$ by~(CL1) and $\rho_\cc(\cl(V))\leq\rho_\cc(\F^n)-i/m$ thanks to \eqref{e-ClosureRho}.
All of this implies the desired inequality $b_i^\cc(\cC)\leq a_i^\cc(\cC)$ for all~$i$.
The last statement for $i=1$ follows from \cref{C-drkAntiCode}.
\end{proof}

Our final result allows us to determine the rank distance of a rank-metric code via the hyperplanes of the associated \qPM.

\begin{cor}
Let $\cC\leq\Fnm$ be a  code.
\begin{alphalist}
\item If $m>n$, then $\drk(\cC)=n-\max\{\dim H\mid H\text{ is a hyperplane in }\cM_\cc(\cC)\}$.
\item If $m=n$, then $\drk(\cC)=n-w$, where
\[
      w=\max\{\dim H\mmid H\text{ is a hyperplane in }\cM_\cc(\cC)\}
                    =\max\{\dim H\mmid H\text{ is a hyperplane in }\cM_\rr(\cC)\}.
\]
\end{alphalist}
\end{cor}

\begin{proof}
It is well known \cite[Thm.~30]{Ra16b} that $\drk(\cC)=a_1(\cC)$.
Suppose first that $m>n$. Denote by~$\cF_\cc$ the collection of flats of the \qPM{} $\cM_\cc(\cC)$.
\cref{T-GenWeightqPMFlats} implies $\drk(\cC)=n-v$, where
$v=\max\{\dim V\mid V\in\cF_\cc,\,\rho_\cc(V)\leq\rho_\cc(\F^n)-1/m\}$.
Let $V\in\cF_\cc$ be such that $\dim V=v$ and $\rho_\cc(V)\leq\rho_\cc(\F^n)-1/m$.
Suppose~$V$ is not a hyperplane.
Then there exists $V'\in\cF_\cc$ such that $V\lneq V'\lneq \F^n$.
Since all these spaces are flats we have $\rho_\cc(V)<\rho_\cc(V')<\rho_\cc(\F^n)$.
Using that~$m$ is a denominator of $\cM_\cc(\cC)$, this yields $\rho_\cc(V')\leq \rho_\cc(\F^n)-1/m$.
Now $\dim V'>v$ leads to a contradiction. Thus~$V$ is a hyperplane.
This concludes the case $m>n$.
For $m=n$ we know from \cref{C-drkAntiCode} that $\drk(\cC)=a_1^\cc(\cC)=a_1^\rr(\cC)$.
Thus the result follows from the first part of this proof.
\end{proof}

\section*{Further Questions}
Our study gives rise to plenty of further question. We list a few of them.
\begin{arabiclist}
\item In~\cite{GLJ21I} we derive a cryptomorphic definition of \qPM{}s
        in terms of independent spaces along with their rank values.
        It remains an open question which other cryptomorphic definitions of $q$-matroids~\cite{BCJ22} can be extended to \qPM{}s.
        For instance, is there a cryptomorphism for \qPM{}s based on flats?
        Our examples show that the rank values of the flats need to be taken into account.
        In this context one may also ask whether there is a meaningful notion of cyclic flats.

\item For classical matroids the Tutte polynomial is a crucial invariant.
        It is equivalent to the rank-generating function and, for a representable matroid,
         carries a lot of information about linear block codes; see for instance \cite{Bri10}.
         A rank-generating function for \qPM{}s has been introduced in \cite{Shi19}, where also its relation
         to rank-metric codes is discussed.
         Furthermore, a Tutte polynomial for $q$-matroids has been introduced much earlier in \cite{BCJ17}, but it is not yet fully
         understood which information about an $\F_{q^m}$-linear code~$\cC$ can be extracted from it.
         As of now, it is therefore not clear what the most appropriate notion of a $q$-Tutte polynomial is.
\item In \cref{P-qMatroidExcSpace} we introduced the $q$-analogue of a special case of paving matroid.  While we only made use of it in the context of representability, this generalization gives rise to the question of a meaningful $q$-analogue of paving matroids.

\end{arabiclist}

\section*{Acknowledgement}
We are very grateful to the reviewers for their close reading and many constructive comments. They improved readability of the paper significantly.


\end{document}